\documentclass[11pt,reqno]{amsart}
\usepackage[margin=1in]{geometry}
\usepackage[english]{babel}
\usepackage[utf8]{inputenc}
\usepackage[colorinlistoftodos, color=green!40, prependcaption]{todonotes}
\usepackage{amsthm}
\usepackage{amssymb}
\usepackage{comment}
\usepackage{mathtools}
\usepackage{physics}
\usepackage{xcolor}
\usepackage{graphicx}
\usepackage[export]{adjustbox}
\usepackage{graphbox}
\usepackage{placeins}
\usepackage[T1]{fontenc}
\usepackage{lipsum}
\usepackage{csquotes}
\usepackage[pdftex, pdftitle={Article}, pdfauthor={Author},colorlinks = true, linkcolor = blue, urlcolor  = blue, citecolor = red]{hyperref} 
\usepackage{amsfonts}
\usepackage{amsmath}
\usepackage{stmaryrd}
\usepackage{braket}
\usepackage{dsfont}
\usepackage{bbold}
\usepackage{relsize}
\usepackage{float}
\usepackage[font=scriptsize]{caption}
\let\ACMmaketitle=\maketitle
\renewcommand{\maketitle}{\begingroup\let\footnote=\thanks \ACMmaketitle\endgroup}

\newcommand{\overbar}[1]{\mkern 1.5mu\overline{\mkern-1.5mu#1\mkern-1.5mu}\mkern 1.5mu}
\newcommand{\M}[1]{\mathcal{M}_{#1}}
\newcommand{\St}[1]{\mathcal{S}_{#1}}
\newcommand{\C}[1]{\mathbb{C}^{#1}}

\newcommand{\env}{\operatorname{env}}
\newcommand{\out}{\operatorname{out}}

\newtheorem{theorem}{Theorem}[section]

\newtheorem*{definition*}{Definition}
\newtheorem{proposition}[theorem]{Proposition}
\newtheorem{corollary}[theorem]{Corollary}
\newtheorem{lemma}[theorem]{Lemma}
\newtheorem{remark}[theorem]{Remark}

\newtheorem*{conjecture*}{Conjecture}

\theoremstyle{definition}

\newcommand{\cS}{\mathcal{S}}
\newcommand{\iden}{\mathbb{1}}

\begin{document}
    
\title{Detecting positive quantum capacities of quantum channels}
\author{Satvik Singh}
\email{satviksingh2@gmail.com}
\address{\parbox{\linewidth}{Department of Applied Mathematics and Theoretical Physics, \\ University of Cambridge, Cambridge, United Kingdom }}

\author{Nilanjana Datta}
\email{n.datta@damtp.cam.ac.uk}
\address{\parbox{\linewidth}{Department of Applied Mathematics and Theoretical Physics, \\ University of Cambridge, Cambridge, United Kingdom }}

\begin{abstract}

Determining whether a noisy quantum channel can be used to reliably transmit quantum information at a non-zero rate is a challenging problem in quantum information theory. This is because it requires computation of the channel's coherent information for an unbounded number of copies of the channel. In this paper, we devise an elementary perturbative technique to solve this problem in a wide variety of circumstances. Our analysis reveals that a channel's ability to transmit information is intimately connected to the relative sizes of its input, output, and environment spaces. We exploit this link to develop easy tests which can be used to detect positivity of quantum channel capacities simply by comparing the channels' input, output, and environment dimensions. Several noteworthy examples, such as the depolarizing and transpose-depolarizing channels (including the Werner-Holevo channel), dephasing channels, generalized Pauli channels, multi-level amplitude damping channels, and (conjugate) diagonal unitary covariant channels, serve to aptly exhibit the utility of our method. Notably, in all these examples, the coherent information of a single copy of the channel turns out to be positive.

\end{abstract}

\maketitle

{\small \begin{center}
     \emph{Dedicated to the thousands of people who have lost their lives to the pandemic -- especially in India during the last few weeks -- including Raj Datta, a beloved brother who lost his battle on 11.05.21.}
\end{center}}

\tableofcontents

\section{Introduction}\label{sec:intro}
The capacity of a noisy communication channel quantifies the fundamental physical limit on noiseless communication through it. Shannon proved that every noisy {\em{classical channel}} has a unique capacity which is given in terms of the mutual information between the random variables characterizing the channel's input and output \cite{Shannon1948communication}.  In contrast, a {\em{quantum channel}} has many different kinds of capacities. These depend on a variety of factors, for example, on the type of information (classical or quantum) which is being transmitted, whether or not this information is private, the nature of the
input states (entangled or not), and whether any auxiliary resource is available to assist the transmission. Auxiliary
resources like prior shared entanglement between the sender and the receiver can enhance the capacities of a quantum
channel. This is in contrast to the case of a classical channel where auxiliary resources, such as shared randomness
between the sender and the receiver, fail to enhance the capacity.

The {\em{quantum capacity}} of a quantum channel quantifies the maximum rate (measured in bits per use of the channel) of noiseless and coherent quantum communication through it, in the limit of asymptotically many uses of the channel and in the absence of any auxiliary resource. By the seminal works of Lloyd~\cite{Lloyd1997capacity}, Shor~\cite{Shor2007capacity}, and Devetak~\cite{Devetak2005capacity}, we know that the quantum capacity of a channel $\Phi$ is given by the following {\em{regularized}} entropic expression:
\begin{align}\label{eq:Qcapacity}
    \mathcal{Q}(\Phi) = \lim_{n \to \infty} \frac{\mathcal{Q}^{(1)}(\Phi^{\otimes n})}{n},
    \end{align}
    where
    \begin{align}
    \mathcal{Q}^{(1)}(\Phi) := \max_\rho I_c(\rho;\Phi) \quad\text{and}\quad I_c(\rho;\Phi) :=  S[\Phi(\rho)]- S[\Phi_c(\rho)].
\end{align}
Here, $\Phi_c$ denotes a channel which is complementary to $\Phi$ (see Section~\ref{subsec:def} for the precise definition), and for any quantum state $\rho$, $S(\rho):= -\Tr (\rho \log \rho)$ denotes its von Neumann entropy. A channel's complement models the loss of information by the channel to the environment and hence crucially affects the channel's ability to transmit information. The quantity $ \mathcal{Q}^{(1)}(\Phi)$ is called the {\em{coherent information}} of $\Phi$ and provides a lower bound to the quantum capacity: $\mathcal{Q}(\Phi)\geq\mathcal{Q}^{(1)}(\Phi)$. This is because $\mathcal{Q}^{(1)}(\Phi^{\otimes n})\geq n\mathcal{Q}^{(1)}(\Phi)$ for all $n$. Computing $\mathcal{Q}^{(1)}(\Phi)$ is a challenge in itself, since it involves solving a non-concave optimization problem. On top of it, the coherent information is usually superadditive, i.e. $\mathcal{Q}^{(1)}(\Phi^{\otimes n}) > n\mathcal{Q}^{(1)}(\Phi)$. This means that the regularization in Eq.~\eqref{eq:Qcapacity} is necessary in order to compute the full capacity $\mathcal{Q}(\Phi)$, thus making it notoriously difficult to do so. 

The same difficulties carry over to the task of checking if a given channel has non-zero quantum capacity. In particular, it is known that for any finite $n$, there exist channels $\Phi$ for which $\mathcal{Q}^{(1)}(\Phi^{\otimes n})=0$ yet $Q(\Phi)>0$ \cite{Cubitt2015unbounded}. Furthermore, it has been shown that there exist pairs of quantum channels (say $\Phi_1$ and $\Phi_2$), each of which has zero quantum capacity, but which can be used in tandem to transmit quantum information, i.e.~$\mathcal{Q}(\Phi_1 \otimes \Phi_2)>0$. This startling effect (known as superactivation \cite{Yard2008super}) is a purely quantum phenomenon because classically, if two channels have zero capacity, the capacity of the joint channel must also be zero. Such extreme examples of superadditivity make it extremely non-trivial to detect if a channel has non-zero quantum capacity.

Only two kinds of channels are currently known to have zero quantum capacity, namely PPT and anti-degradable channels \cite{Smith2012incapacity}. Checking if a given channel is PPT is equivalent to checking if the channel's Choi matrix and its partial transpose are positive semidefinite. The task of determining if a given channel is anti-degradable can be modelled as a semidefinite program \cite{sutter2017approx}. However, the question of whether there exist channels that are neither PPT nor anti-degradable but still have zero capacity is wide open. To date, there is no general procedure or algorithm known which can detect if a given channel has positive quantum capacity, except in special circumstances where some numerical techniques can be employed to compute or provide lower bounds on the coherent information \cite{jesse2008unreliable, jesse2008unreliable2}. However, when one wants to detect arbitrarily small positive values of quantum capacities, numerical methods can not only become unreliable, but also computationally expensive in higher dimensions.

In light of the above discussion, it is clear that our understanding of the set of zero capacity quantum channels is still in its nascent phase.
The main result of this paper sheds light on the structure of this set by providing a powerful sufficient condition to guarantee positivity of quantum capacities for a wide variety of quantum channels and their complements; see Theorem~\ref{theorem:main} in the main text, which is also stated below for convenience. Recently, Siddhu has also made interesting progress on the stated problem of quantum capacity detection
\cite{siddhu2020logsingularities}. The ideas employed in \cite{siddhu2020logsingularities} are similar in essence to the key ingredient that we use to obtain our main result. However, our approach is mathematically more rigorous and our results have much wider applicability. Explicit comparisons between our results and those presented in \cite{siddhu2020logsingularities} are made at relevant places in the sequel. \\

\begin{theorem}
Let $\Phi$ and $\Phi_c$ be complementary quantum channels. For a pure state $\ketbra{\psi}$, denote the orthogonal projections onto $\ker \Phi(\ketbra{\psi}{\psi})$ and $\ker \Phi_c(\ketbra{\psi}{\psi})$ by $K_\psi$ and $K^c_\psi$, respectively. Then, 
\begin{itemize}
    \item $\mathcal{Q}(\Phi)\geq\mathcal{Q}^{(1)}(\Phi)>0$ if there exist states $\ketbra{\psi}{\psi}$ and $\sigma$ such that $\operatorname{Tr}(K_\psi \Phi(\sigma))>\operatorname{Tr}(K^c_\psi \Phi_c(\sigma))$.
    \item $\mathcal{Q}(\Phi_c)\geq\mathcal{Q}^{(1)}(\Phi_c)>0$ if there exist states $\ketbra{\psi}{\psi}$ and $\sigma$ such that $\operatorname{Tr}(K_\psi \Phi(\sigma))<\operatorname{Tr}(K^c_\psi \Phi_c(\sigma))$.
\end{itemize}
\end{theorem}

The primary tool employed in the proof of Theorem~\ref{theorem:main} is a seminal result (by Rellich) from the theory of analytic perturbations of Hermitian matrices. Put simply, if a Hermitian matrix $A_0$ with an eigenvalue $\lambda$ of multiplicity $n$ is subjected to a linear (Hermitian) perturbation $A(\epsilon)=A_0 + \epsilon A_1$ in a real parameter $\epsilon$, then the perturbed matrix has exactly $n$ eigenvalues converging to $\lambda$ as $\epsilon\to 0$, all of which, rather remarkably, admit convergent power series expansions in a neighborhood of $\epsilon=0$; see Theorem~\ref{theorem:perturbation}. Now, for a given channel $\Phi$, it is trivial to see that $I_c(\rho;\Phi)=0$ for all pure input states $\rho=\ketbra{\psi}$. It turns out that if we choose $\rho(\epsilon)= (1-\epsilon)\ketbra{\psi}+\epsilon \sigma$ as a slight perturbation of an arbitrarily chosen pure state $\ketbra{\psi}$ with an arbitrary mixed state $\sigma$, then the first order correction terms in the eigenvalue expansions of the perturbed outputs $\Phi[\rho(\epsilon)]$ and $\Phi_c[\rho(\epsilon)]$ can be analyzed to yield a simple condition which guarantees that $I_c(\rho(\epsilon);\Phi)>0$ for sufficiently small values of $\epsilon$, so that $\mathcal{Q}^{(1)}(\Phi)>0 \implies \mathcal{Q}(\Phi)>0$.

An equivalent formulation of the above theorem yields a first-of-its-kind necessary condition for membership in the set of zero capacity quantum channels; see Theorem~\ref{theorem:zero-necessary}. As a sanity check, we prove that some of the few known classes of channels with zero capacity satisfy this condition; see Lemma~\ref{lemma:anti-degaradable}. Notably, the main result in \cite{siddhu2020logsingularities} arises as an immediate consequence of Theorem~\ref{theorem:main}, which can be applied to detect positive quantum capacity of a given channel (or its complement) if the channel's output and environment dimension are unequal and if there exists a pure state whose image under the channel has the maximum possible rank; see Corollary~\ref{corollary:vikesh}. Furthermore, we prove that the aforementioned pure state exists for all channels with sufficiently large input dimensions, so that any such channel (or its complement) with unequal output and environment dimension must have positive quantum capacity, irrespective of the specific form of the channel; see Corollary~\ref{corollary-biginput}. Similar results were obtained in \cite{siddhu2020logsingularities} for channels with large output and environment dimensions, which are derived again in Corollary~\ref{corollary:vikesh2} with simpler proofs. These results constitute a family of simple dimensional tests which can be used to detect positive quantum channel capacities simply by comparing the dimensions of the channels' input, output, and environment spaces. Clearly, our results exhibit an intimate connection between the ability of a channel to transmit information and the relative sizes of the channel’s input, output, and environment spaces

Our methods can be easily applied to several concrete examples of important quantum channels and their complements as well. The following is a list of main results in this direction. We should mention that in all these examples, the coherent information of a single copy of the relevant channel turns out to be positive.

\begin{itemize}
    \item In \cite{Leung2017complementary}, the {qubit depolarizing} channel $\mathcal{D}_p:\M{2}\to \M{2}$ was shown to have positive {complementary quantum capacity} (the complementary quantum capacity of a channel refers to the quantum capacity of anyone of its complements) for all non-zero values of the noise parameter $p>0$. We extend this result to show that the {qudit depolarizing} and {transpose-depolarizing} channels ($\mathcal{D}_p:\M{d}\to \M{d}$ and $\mathcal{D}^\top_q:\M{d}\to \M{d}$, respectively) have positive complementary quantum capacities for all $p>0$ and $q<\frac{d}{d-1}$; see Theorem~\ref{theorem:depol-transp}. Moreover, the {Werner-Holevo} channel $\Phi_{\rm WH}\equiv \mathcal{D}^\top_{\frac{d}{d-1}}:\M{d}\to \M{d}$ is shown to have positive complementary quantum capacity for all $d\geq 4$; see Theorem~\ref{theorem:werner-holevo}.
    
    \item In \cite{Leung2017complementary}, the {qubit Pauli} channel $\Phi_P:\M{2}\to \M{2}$ was also shown to have positive complementary quantum capacity whenever the defining probability matrix $P\in\M{2}$ has at least three non-zero entries. An extension of this result for generalized {qudit Pauli} channels $\Phi_P:\M{d}\to \M{d}$ with a much simpler proof is presented in Theorem~\ref{theorem:pauli}.
    
    \item 
    For a {multi-level amplitude damping} channel $\Phi_{\Vec{\gamma}}:\M{d}\to \M{d}$, we establish simple constraints on the decay rate vector $\Vec{\gamma}\in\mathbb{R}^{d(d-1)/2}$ which ensure the positivity of its quantum capacity and its complementary quantum capacity; see Theorem~\ref{theorem:MAD}.
    
    \item We show that a generalized {dephasing} or {Hadamard} channel $\Phi_B:\M{d}\to \M{d}$ (parametrized by a {correlation} matrix $B\in\M{d}$) has zero quantum capacity if and only if it is {entanglement-breaking}; see Theorem~\ref{theorem:dephasing}.
    
    \item Recently, the family of {(conjugate) diagonal unitary covariant} (dubbed (C)DUC) quantum channels was introduced in~\cite{singh2020diagonal}. A rich variety of channels, such as the depolarizing and transpose depolarizing channels, amplitude damping channels, dephasing channels etc.~are known to belong in this family. In Proposition~\ref{prop:CDUC-rank}, we completely characterize the class of (C)DUC channels for which there exists a pure input state which gets mapped to a maximal rank output state, so that Corollary~\ref{corollary:vikesh} can be applied to infer positivity of the quantum capacities of these channels and their complements; see Theorem~\ref{theorem:CDUC-cap}.  
\end{itemize}

Moreover, Theorem~\ref{theorem:main} has various interesting ramifications, which are studied in Section~\ref{sec:morecapable}. For instance, it leads to simplified proofs of certain existing structure theorems for the class of {degradable} quantum channels, and an extension of their applicability to the larger class of {more capable} quantum channels.

Before wrapping up the introduction, we must emphasize that this paper is {not} targeted towards the problem of computing lower bounds on quantum channel capacities. While this is certainly an interesting question to think about, here we are interested in something even more fundamental: When is the quantum capacity of any given channel non-zero? Tackling this basic yes/no question is the primary aim of our work. We do so by devising a simple perturbative technique to detect positivity (no matter how small it might be) of quantum channel capacities.

\section{Prerequisites} \label{subsec:def}

In this section, we briefly review the basics of quantum channels and their quantum capacities. A more thorough discussion on these topics can be found in \cite{watrous2018theory}.


We denote the set of all $d_1\times d_2$ complex matrices by $\M{d_1\times d_2}$. When $d_1=d_2=d$, we denote the corresponding matrix space by $\M{d}:=\M{d\times d}$. For $A,B\in \M{d}$, $A\leq B$ means that $B-A$ is positive semi-definite. The convex set of {quantum states} (positive semi-definite matrices with unit trace) in $\M{d}$ and $\M{d}^{\otimes n}$ (for $n\in\mathbb{N}$) will be denoted by $\St{d}$ and $\St{d^{\otimes n}}$, respectively. 
The {vectorization} map $\operatorname{vec}: \M{d_1\times d_2}\to \C{d_1}\otimes \C{d_2}$ is a linear bijection defined as
\begin{equation}\label{eq:vec}
    \forall X\in\M{d_1\times d_2}: \quad \operatorname{vec}X = \sum_{i,j=0}^{d-1} X_{ij}\ket{ij},
\end{equation}
where $\{\ket{i}\}_{i\in [d]}$ denotes the standard orthonormal basis of $\C{d}$ and $[d]:=\{0,1,\ldots ,d-1\}$. For any given $\ket{\psi}\in\C{d_1}\otimes \C{d_2}$, one can use the singular value decomposition of $\operatorname{vec}^{-1}\ket{\psi}\in\M{d_1\times d_2}$ to obtain the {Schmidt decomposition} of $\ket{\psi}=\sum_{i=0}^{n-1} \sqrt{s_i}\ket{i}_1\ket{i}_2$, where $\{\ket{i}_1\}_{i\in [n]}\subseteq \C{d_1}$ and $\{\ket{i}_2\}_{i\in [n]}\subseteq \C{d_2}$ are orthonormal sets, $s_i\geq 0$ and $n=\min\{d_1,d_2\}$. The number of non-zero $s_i$ is called the {Schmidt rank} of $\ket{\psi}$, and is equal to $\operatorname{rank}(\operatorname{vec}^{-1}\ket{\psi}) \leq \min\{d_1,d_2\}$.

A {quantum channel} is a completely positive and trace preserving linear map $\Phi : \M{d_1}\rightarrow \M{d_2}$. $\Phi$ is said to be {unital} if $\Phi({\iden}_{d_1})= {\iden}_{d_2}$, where ${\iden}_d\in\M{d}$ is the identity matrix. Its {adjoint} $\Phi^*:\M{d_2}\rightarrow\M{d_1}$ is a unital completely positive linear map defined uniquely by the following relation: 
\begin{equation}
\forall X\in\M{d_1}, \forall Y\in\M{d_2}: \quad \operatorname{Tr}[\Phi(X)Y]=\operatorname{Tr}[X\Phi^*(Y)]. 
\end{equation}
Two quantum channels $\Phi:\M{d}\rightarrow \M{d_{\out}}$ and $\Phi_c:\M{d}\rightarrow \M{d_{\operatorname{env}}}$ are {complementary} to each other if there exists an isometry (often called a {Stinespring} isometry) $V:\C{d}\rightarrow \C{d_{\out}}\otimes \C{d_{\operatorname{env}}}$ such that 
\begin{equation}\label{eq:complementary}
\forall X\in \M{d}: \quad\Phi(X)=\operatorname{Tr}_{\operatorname{env}}(VXV^\dagger) \quad\text{and}\quad
\Phi_c (X) = \operatorname{Tr}_{\out}(VXV^\dagger).
\end{equation}  
In the above scenario, all input states are first isometrically embedded into the combined output-environment space, with the action of $\Phi$ and $\Phi_c$ then retrieved by partially tracing out the environment and the output space, respectively \cite{King2007complement,Holevo2007complementary}. For a given channel $\Phi: \M{d}\rightarrow \M{d_{\out}}$, the collection of all channels that are complementary to $\Phi$ is denoted by $\mathcal{C}_{\Phi}$. 

\begin{remark}\label{remark:iso-exten}
To every channel $\Phi:\M{d}\to\M{d_{\out}}$, we can uniquely associate two sets of channels $\mathcal{C}_\Phi$ and $\mathcal{C}_{\Phi_c}$, where $\Phi_c\in\mathcal{C}_\Phi$ is complementary to $\Phi$. These sets are non-empty, since the Stinespring dilation theorem guarantees that $\Phi$ can be expressed as in Eq.~\eqref{eq:complementary} for some isometry $V:\C{d}\rightarrow \C{d_{\out}}\otimes \C{d_{\operatorname{env}}}$ \cite[Corollary 2.27]{watrous2018theory}. Furthermore, different choices of $\Phi_c\in\mathcal{C}_\Phi$ yield the same set $\mathcal{C}_{\Phi_c}$. Finally, it is straightforward to show that 
\begin{itemize}
    \item any two channels $\Phi_1,\Phi_2\in \mathcal{C}_\Phi$ (or $\mathcal{C}_{\Phi_c}$) with $\Phi_1:\M{d}\to \M{d_1}$, $\Phi_2:\M{d}\to \M{d_2}$, and $d_1\leq d_2$ are {isometrically related}, i.e. $\exists$ an isometry $W:\C{d_1}\to \C{d_2}$ such that $\Phi_2 (X) = W\Phi_1 (X) W^\dagger$.
    \item any two channels $\Phi_1\in\mathcal{C}_\Phi$ and $\Phi_2\in\mathcal{C}_{\Phi_c}$ are complementary to each other.
\end{itemize}
\end{remark}  

We define the {minimal environment} and {output dimension} of $\Phi$ as
\begin{align}
    d^{*}_{\operatorname{env}}(\Phi) :=  \min \{d_{\operatorname{env}}: \exists\Phi_c\in \mathcal{C}_{\Phi}, \Phi_c:\M{d}\rightarrow \M{d_{\env}} \} \quad\text{and}\quad d^{*}_{\out} (\Phi) := d^{*}_{\env}(\Phi_c),
\end{align}
respectively, where $\Phi_c\in\mathcal{C}_\Phi$. It is clear from Remark~\ref{remark:iso-exten} that the above definition is independent of the choice of $\Phi_c$, see also Lemma~\ref{lemma:minimal}. We say that two complementary channels $\Phi:\M{d}\to\M{d_{\out}}$ and $\Phi_c:\M{d}\to\M{d_{\env}}$ are {minimally defined} if $d^*_{\env}(\Phi)=d_{\env}$ and $d^*_{\out}(\Phi)=d_{\out}$. If two complementary channels $\Phi$ and $\Phi_c$ are minimally defined, then they are isometrically related to every channel in the sets $\mathcal{C}_{\Phi_c}$ and $\mathcal{C}_{\Phi}$, respectively, in the sense of Remark~\ref{remark:iso-exten}. 

\begin{remark}
Intuitively, one can think about the minimal output dimension of $\Phi:\M{d}\to \M{d_{\out}}$ as the minimal size of the output space that can accommodate all the channel outputs. More precisely, we show in Lemma~\ref{lemma:minimal} that 
\begin{equation}
    \forall \rho\in\St{d}: \quad \operatorname{range}\Phi(\rho) \subseteq \operatorname{range}\Phi(\iden_d).
\end{equation}
Thus, even though $\Phi$ is originally defined with an output space of dimension $d_{\out}$, a smaller size $d^*_{\out}(\Phi)=\dim \operatorname{range}\Phi(\iden_d)$ actually suffices to fully accommodate the output from $\Phi$. 
\end{remark}

The {\em{Choi matrix}} of a quantum channel $\Phi:\M{d}\rightarrow \M{d_{\out}}$ is defined as~\cite{Choi1975iso, Jamiokowski1972iso}
\begin{align}\label{eq:choi}
    J(\Phi) &:= ( \Phi \otimes {\rm{id}})\ketbra{\Omega}, \quad\text{where}\quad \ket{\Omega} := \sum_{i=0}^{d-1} \ket{i}\ket{i} \, \in {\mathbb{C}}^d \otimes {\mathbb{C}}^d
\end{align}
denotes a maximally entangled state and ${\rm{id}}:\M{d}\rightarrow \M{d}$ is the identity map. The rank of the Choi matrix of a channel is called the {Choi rank} of the channel. The minimal environment dimension of a channel is related to its Choi rank via the following crucial lemma. 
\begin{lemma}\label{lemma:minimal}
For a channel $\Phi:\M{d}\rightarrow \M{d_{\out}}$ and some complementary channel $\Phi_c\in \mathcal{C}_\Phi$,
\begin{alignat}{2}
    d^{*}_{\operatorname{env}}(\Phi)&=\operatorname{rank}J(\Phi)&&=\operatorname{rank}\Phi_c(\iden_d) \\
    d^{*}_{\out} (\Phi)&=\operatorname{rank}J(\Phi_c)&&=\operatorname{rank}\Phi(\iden_d).
\end{alignat}
\end{lemma}
\begin{proof}
See the Methods section~\ref{appen:lemmaminimal}.
\end{proof}

\begin{remark}\label{remark:minimal}
We note some straightforward consequences of Lemma~\ref{lemma:minimal} below:
\begin{itemize}
    \item If $\Phi:\M{d}\rightarrow \M{d_{\out}}$ is a unital channel, then $d^{*}_{\out}(\Phi)= d_{\out}$.
    \item If a channel $\Phi:\M{d}\rightarrow \M{d_{\out}}$ admits a Kraus representation $\Phi(X)=\sum_{i=1}^k A_i X A_i^\dagger$, where $\{A_i\}_{i=1}^k$ is a linearly independent set of $d_{\out}\times d$ matrices, then $d^*_{\env}(\Phi)=k$.
    \item For a channel $\Phi:\M{d}\to\M{d_{\out}}$, its adjoint $\Phi^*$ has $d^*_{\env}(\Phi^*)=d^*_{\env}(\Phi)$ and $d^*_{\out}(\Phi^*)=d$.
\end{itemize} 
\end{remark}

The {coherent information} of a state $\rho\in\St{d}$ with respect to a channel $\Phi:\M{d}\rightarrow \M{d_{\out}}$ is defined as $I_c(\rho; \Phi) \coloneqq S[\Phi(\rho)] - S[\Phi_c(\rho)]$, where $\Phi_c\in\mathcal{C}_\Phi$ is complementary to $\Phi$ and $S(\rho)=-\operatorname{Tr}(\rho\log \rho)$ denotes the {von Neumann entropy} of $\rho\in\St{d}$. The {coherent information} of $\Phi$ is defined to be 
\begin{equation}\label{eq:coherent}
    \mathcal{Q}^{(1)}(\Phi) := \max_{\rho\in\St{d}} I_c(\rho;\Phi).
\end{equation}
The {quantum capacity} of $\Phi$ admits the following regularized expression:
\begin{equation}\label{eq:capacity}
    \mathcal{Q}(\Phi) = \lim_{n\to \infty} \frac{\mathcal{Q}^{(1)}(\Phi^{\otimes n})}{n},
\end{equation}
where, for $n\in\mathbb{N}$, the coherent information of the product channel $\mathcal{Q}^{(1)}(\Phi^{\otimes n})$ is called the $n-${shot coherent information} of $\Phi$. Clearly, $\mathcal{Q}^{(1)}(\Phi^{\otimes n})\geq n\mathcal{Q}^{(1)}(\Phi)$. However, it may happen that the coherent information is {superadditive}: $\mathcal{Q}^{(1)}(\Phi^{\otimes n})>n\mathcal{Q}^{(1)}(\Phi)$, in which case the task of evaluating $\mathcal{Q}(\Phi)$ becomes intractable \cite{smolin1998noisy, Cubitt2015unbounded}. The regularization in Eq.~\eqref{eq:capacity} is not required only for certain special channels for which the coherent information is additive: $\mathcal{Q}^{(1)}(\Phi^{\otimes n})=n\mathcal{Q}^{(1)}(\Phi)$. Consequently, the quantum capacities of these channels admit nice single-letter expressions in terms of the channels' coherent information $\mathcal{Q}(\Phi)=\mathcal{Q}^{(1)}(\Phi)$. 

\begin{remark}\label{remark:iso-cap}
Since for a given channel $\Phi$, all complementary channels $\Phi_c\in\mathcal{C}_\Phi$ are isometrically related (in the sense of Remark~\ref{remark:iso-exten}), the stated capacity expressions do not depend on the choice of $\Phi_c$. More generally, the quantum capacities of any two isometrically related channels are identical. Hence, while computing quantum capacities, one can choose to work with channels $\Phi:\M{d}\to \M{d_{\out}}$ and $\Phi_c:\M{d}\to\M{d_{\env}}$ that are minimally defined (i.e. $d_{\env}=d^*_{\env}(\Phi)$ and $d_{\out}=d^*_{\out}(\Phi)$).
\end{remark}  


Degradable channels are the quintessential examples of channels which have additive coherent information \cite{Devetak2005degradable}. A quantum channel $\Phi:\M{d}\rightarrow\M{d_{\out}}$ is said to be {degradable} if there exists a channel $\mathcal{N}:\M{d_{\out}}\rightarrow\M{d_{\operatorname{env}}}$ such that $\Phi_c=\mathcal{N}\circ \Phi$ for some complementary channel $\Phi_c:\M{d}\to\M{d_{\env}}$. Quantum channels that are complementary to degradable channels are known as {anti-degradable}. In other words, a channel $\Phi:\M{d}\to\M{d_{\out}}$ is said to be {anti-degradable} if there exists a channel $\mathcal{N}:\M{d_{\env}}\to \M{d_{\out}}$ such that $\Phi = \mathcal{N}\circ\Phi_c$ for some complementary channel $\Phi_c:\M{d}\to\M{d_{\env}}$. If such a channel could be used for reliable quantum communication, then its environment would be able to replicate all the transmitted information by applying the anti-degrading map $\mathcal{N}$, thus violating the {no-cloning} theorem \cite{Bennett1997anti,Smith2012incapacity}. Hence, anti-degrading channels have zero quantum capacity. More generally, all the currently known classes of channels with additive coherent information lie within the strict superset of {more capable} channels \cite{Watanabe2012capable} (see Figure~\ref{fig:1}), which are defined by the property that their complementary channels have zero quantum capacity.

\begin{figure}
    \centering
    \includegraphics[scale=0.6]{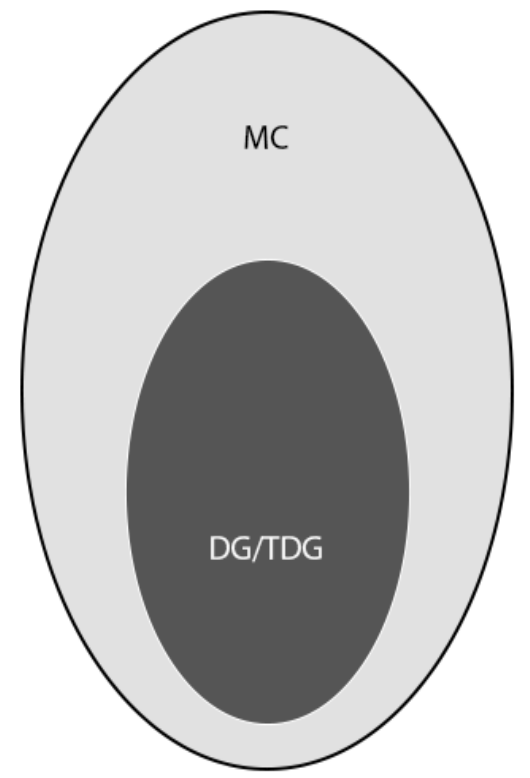}
    \caption{\cite{Watanabe2012capable} The set of more capable (MC) quantum channels, i.e. channels with zero complementary quantum capacity. Degradable (DG) and transpose degradable (TDG) channels lie strictly within this set. The relationship between the sets of DG and TDG channels is currently unestablished.}
    \label{fig:1}
\end{figure}

\begin{figure}
    \centering
    \includegraphics[scale=0.5]{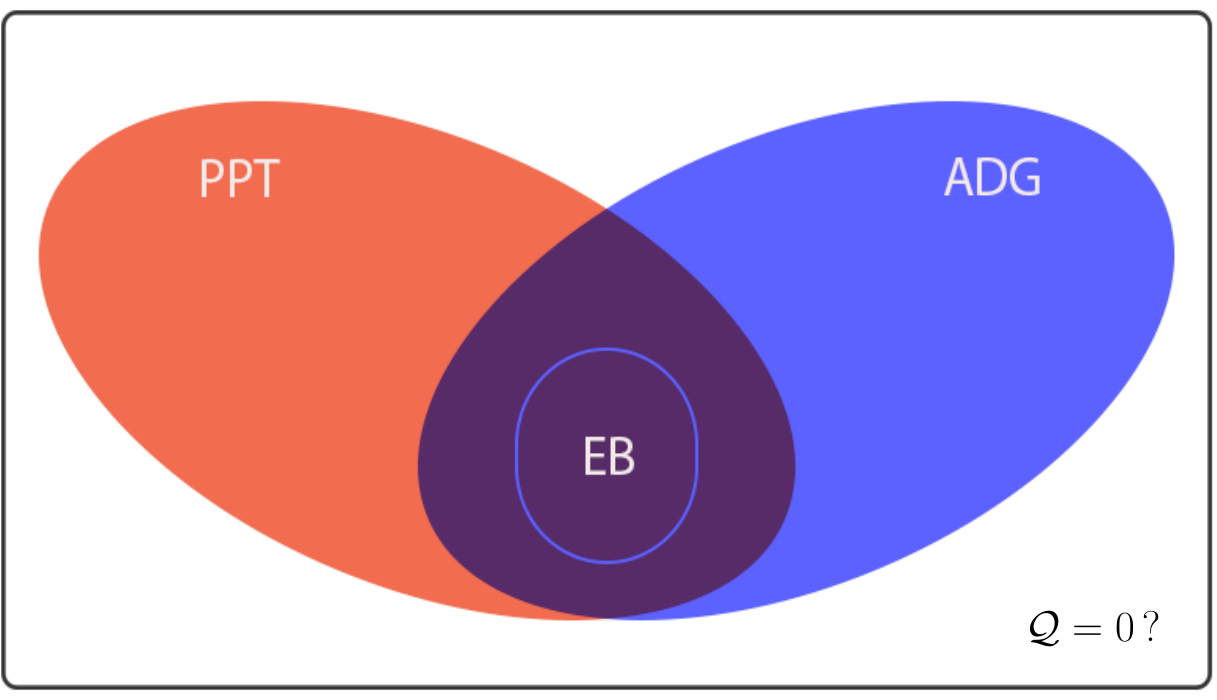}
    \caption{ \cite{Smith2012incapacity} The set of quantum channels with zero quantum capacity: PPT and anti-degradable (ADG) channels are the only kinds of channels which are currently known to belong in this set. Entanglement-breaking (EB) channels form a strict subset of the intersection of the PPT and anti-degradable sets of channels. The question mark indicates that it is not currently known if there exist quantum channels which are neither anti-degradable nor PPT but still have zero quantum capacity.}
    \label{fig:zero-capacity}
\end{figure}

Apart from anti-degradable channels, only PPT channels \cite{Horodecki2000binding} are known to have zero capacity. A channel $\Phi:\M{d}\to \M{d_{\out}}$ is said to be PPT if $\top\circ\Phi$ is again a quantum channel, where $\top:\M{d_{\out}}\to \M{d_{\out}}$ is the transpose map. Equivalently, $\Phi$ is PPT if and only if its Choi matrix $J(\Phi)\in \M{d_{\out}}\otimes \M{d}$ and its partial transposite are positive semi-definite. Such channels cannot have positive capacity because if they did, one would be able to distill maximally entangled pure states from the PPT Choi matrices of these channels, which is impossible \cite{Horodecki1998bound}. The well-known family of {entanglement-breaking} channels \cite{Horodecki2003entanglement} -- which consists of channels whose Choi matrices are {separable} -- is strictly contained within the the intersection of the anti-degradable and PPT families, see Figure~\ref{fig:zero-capacity}.

\begin{remark}
In \cite{Brdler2010conjugate}, the authors introduced another class of quantum channels with additive coherent information, namely {conjugate degradable} channels. According to their definition, a channel $\Phi:\M{d}\to \M{d_{\out}}$ is conjugate degradable if there exists a channel $\mathcal{N}:\M{d_{\out}}\to \M{d_{\env}}$ such that
\begin{equation}\label{eq:conj}
\mathcal{C} \circ \Phi_c = \mathcal{N} \circ \Phi,
\end{equation}
where $\mathcal{C}$ denotes (entrywise) complex conjugation on $\M{d_{\env}}$ and $\Phi_c:\M{d}\to\M{d_{\env}}$ is complementary to $\Phi$. In an analogous fashion, the class of {conjugate anti-degradable} channels was defined, and these channels were claimed to have zero quantum capacity. However, there is a fundamental problem with the stated definitions, which stems from the fact that the map $\mathcal{C}$ is {anti-linear}. Therefore, Eq.~\eqref{eq:conj} is invalid, since it is equating an anti-linear map on its left hand side with a linear map on its right hand side. This issue can be easily resolved if we just replace the anti-linear operation of complex conjugation by the {linear} operation of transposition. Hence, we are led to the following definitions of transpose degradable and transpose anti-degradable channels.
\end{remark}

 We say that a channel $\Phi:\M{d}\to\M{d_{\out}}$ is {transpose degradable} (resp. {transpose anti-degradable}) if there exists a channel $\mathcal{N}:\M{d_{\out}}\rightarrow\M{d_{\operatorname{env}}}$ (resp. $\mathcal{N}:\M{d_{\env}}\rightarrow\M{d_{\operatorname{out}}}$) such that $\top\circ \Phi_c=\mathcal{N}\circ \Phi$ (resp. $\top\circ\Phi = \mathcal{N}\circ\Phi_c$) for some complementary channel $\Phi_c:\M{d}\to\M{d_{\env}}$, where $\top$ denotes the transpose map on the relevant matrix spaces. It is easy to show that transpose degradable channels have additive coherent information, while transpose anti-degradable channels have zero quantum capacity. In this regard, note that all transpose anti-degradable channels are PPT as well. It would be interesting to see if these classes are actually different from their non-transposed counterparts.

\section{Main results}\label{sec:main}

Let $\Phi:\M{d}\to \M{d_{\out}}$ be a quantum channel and let $\Phi_c\in\mathcal{C}_{\Phi}$ be a complementary channel. Our goal is to check if $\mathcal{Q}(\Phi)>0$. Since $\mathcal{Q}(\Phi)\geq \mathcal{Q}^{(1)}(\Phi)$, we proceed by checking if $\mathcal{Q}^{(1)}(\Phi)>0$. For a pure input state $\ketbra{\psi}\in \St{d}$, it is easy to show that $I_c(\ketbra{\psi};\Phi)=0$, since the non-zero eigenvalues of $\Phi(\ketbra{\psi})$ and $\Phi_c(\ketbra{\psi})$ (counted with multiplicities) are identical, see Lemma~\ref{lemma:2}. Now, the idea is to cleverly choose a pure input state and perturb it along the direction of a suitable mixed state in order to obtain a positive value for the coherent information. With this end in sight, let us break the technical aspects of this idea into the following steps:

\begin{itemize}
    \item For a pure input state $\ketbra{\psi}$ and a mixed state $\sigma$, define the one-parameter family of states $\rho(\epsilon) = (1-\epsilon)\ketbra{\psi} + \epsilon \sigma,$ where $\epsilon\in [0,1]$ is the perturbation parameter.
    \item Focus on the zero eigenvalue of $\Phi(\ketbra{\psi})$ and $\Phi_c(\ketbra{\psi})$ with respective multiplicities 
    $\kappa=\dim \ker\Phi(\ketbra{\psi})$ and $\kappa_c=\dim\ker\Phi_c(\ketbra{\psi}).$
    \item Once the perturbation is turned on,
\begin{equation}
    \Phi[\rho(\epsilon)] = (1-\epsilon)\Phi(\ketbra{\psi}) + \epsilon \Phi(\sigma) \quad\text{and}\quad \Phi_c[\rho(\epsilon)] = (1-\epsilon)\Phi_c(\ketbra{\psi}) + \epsilon \Phi_c(\sigma)
\end{equation}
will have exactly $\kappa$ and $\kappa_c$ eigenvalues, respectively, which converge to zero as $\epsilon\to 0$. Remarkably, these eigenvalues admit convergent power series expansions in the perturbation parameter $\epsilon$ (in a neighborhood of $\epsilon=0$), see Theorem~\ref{theorem:perturbation}. Furthermore, if $K_\psi$ and $K^c_\psi$ denote the orthogonal projections onto the unperturbed eigenspaces $\ker\Phi(\ketbra{\psi})$ and $\ker\Phi_c(\ketbra{\psi})$, respectively, then the non-zero eigenvalues of $K_\psi\Phi(\sigma)K_\psi$ and $K^c_\psi\Phi(\sigma)K^c_\psi$ determine the first order correction constants in the aforementioned eigenvalue expansions, thus giving rise to the following crucial term in the derivative of the coherent information:
\begin{equation}\label{eq:log-singularity}
    I'_c(\rho(\epsilon);\Phi) = \left[ \operatorname{Tr}(K_\psi\Phi(\sigma)) - \operatorname{Tr}(K^c_\psi\Phi_c(\sigma)) \right] \log (\frac{1}{\epsilon}) + \ldots
\end{equation}
\item All the higher order corrections above can be shown to be bounded in $\epsilon$, so that the derivative becomes positive for sufficiently small values of $\epsilon$ if $\operatorname{Tr}(K_\psi\Phi(\sigma))>\operatorname{Tr}(K^c_\psi\Phi_c(\sigma))$. Hence, $I_c(\rho(\epsilon);\Phi)$ is strictly increasing when $0<\epsilon<\delta$ (for some small $\delta > 0$), which implies that it has a positive value in the stated range (recall that at $\epsilon=0$, $I_c(\rho(\epsilon);\Phi)=0$).
\end{itemize}

\begin{remark}
    In \cite[Sections III and IV]{siddhu2020logsingularities}, the first order corrections in the  eigenvalues of $\Phi[\rho(\epsilon)]$ and $\Phi_c[\rho(\epsilon)]$ which give rise to the trace factor in Eq.~\eqref{eq:log-singularity} are referred to as the {rates} of the $\epsilon${-log singluarities} of the output entropies $S[\Phi(\rho(\epsilon))]$ and $S[\Phi_c(\rho(\epsilon))]$.
\end{remark}

In order to formulate the above discussion in a mathematically rigorous fashion, a number of non-trivial intermediate steps are required, which are presented in the Methods section \ref{appendix:main}. The final theorem stated below is the primary result of this paper.

\begin{theorem} \label{theorem:main}
Let $\Phi:\M{d}\rightarrow \M{d_{\out}}$ and $\Phi_c:\M{d}\to \M{d_{\env}}$ be complementary channels. For a pure input state $\ketbra{\psi}{\psi}\in \St{d}$, denote the orthogonal projections onto $\ker \Phi(\ketbra{\psi}{\psi})$ and $\ker \Phi_c(\ketbra{\psi}{\psi})$ by $K_\psi$ and $K^c_\psi$, respectively. Then,
\begin{itemize}
    \item $\mathcal{Q}(\Phi)\geq\mathcal{Q}^{(1)}(\Phi)>0$ if there exist states $\ketbra{\psi}{\psi},\sigma\in\St{d}$ such that
\begin{equation}
\operatorname{Tr}(K_\psi \Phi(\sigma))>\operatorname{Tr}(K^c_\psi \Phi_c(\sigma)). 
\end{equation}
    \item $\mathcal{Q}(\Phi_c)\geq\mathcal{Q}^{(1)}(\Phi_c)>0$ if there exist states $\ketbra{\psi}{\psi},\sigma\in\St{d}$ such that 
    \begin{equation}
        \operatorname{Tr}(K_\psi \Phi(\sigma))<\operatorname{Tr}(K^c_\psi \Phi_c(\sigma)).
    \end{equation}
\end{itemize}
\end{theorem}

Before proceeding further, let us note  that for a given channel $\Phi:\M{d}\to\M{d_{\out}}$ and an arbitrary pure state $\ketbra{\psi}\in \St{d}$, Lemma~\ref{lemma:minimal} and Lemma~\ref{lemma:2} can be exploited to deduce that the maximal possible rank of $\Phi(\ketbra{\psi})$ is $\min\{d^*_{\out}(\Phi),d^*_{\env}(\Phi)\}$. In other words,
\begin{equation}\label{eq:max-pure-rank}
\forall \ketbra{\psi}\in\St{d}: \quad \operatorname{rank}\Phi(\ketbra{\psi})\leq \min\{d^*_{\out}(\Phi),d^*_{\env}(\Phi)\}.
\end{equation} 
It is not easy to determine if there exists a pure state for which the above bound gets saturated. However, if we assume that such a pure state exists, a simple application of the previous theorem yields a very useful corollary, which we now state and prove. We should mention that this result has been derived independently in~\cite{siddhu2020logsingularities}.

\begin{corollary}\label{corollary:vikesh}
Let $\Phi:\M{d}\rightarrow \M{d_{\out}}$ be a channel such that there exists a pure state $\ketbra{\psi}\in\St{d}$ with $\operatorname{rank}\Phi(\ketbra{\psi})= \min\{d^*_{\out}(\Phi),d^*_{\env}(\Phi)\}$ and $\Phi_c\in\mathcal{C}_\Phi$. Then,
\begin{itemize}
    \item $\mathcal{Q}(\Phi)\geq\mathcal{Q}^{(1)}(\Phi)>0 \,\text{ if }\, d^{*}_{\out}(\Phi)>d^{*}_{\operatorname{env}}(\Phi)$.
    
    \item $\mathcal{Q}(\Phi_c)\geq\mathcal{Q}^{(1)}(\Phi_c)>0 \,\text{ if }\, d^{*}_{\out}(\Phi)<d^{*}_{\operatorname{env}}(\Phi).$
\end{itemize}
\end{corollary}
\begin{proof}
We only prove the second part here and leave a similar proof of the first part to the reader. Firstly, without loss of generality, we can assume that $\Phi:\M{d}\to\M{d_{\out}}$ and $\Phi_c:\M{d}\to\M{d_{\env}}$ are minimally defined, so that $d^*_{\out}(\Phi)=d_{\out}$ and $d^*_{\env}(\Phi)=d_{\env}$ (see Remark~\ref{remark:iso-cap}). Now, if there exists $\ketbra{\psi}{\psi}\in \mathcal{S}_d$ with $\operatorname{rank}\Phi(\ketbra{\psi}{\psi})=d_{\out}$, then $\ker\Phi(\ketbra{\psi}{\psi})=\{0\}$. Moreover, since $d_{\out}<d_{\operatorname{env}}$, Lemma~\ref{lemma:2} tells us that $\ker\Phi_c(\ketbra{\psi}{\psi})\neq \{0\}$. Hence, in the notation of Theorem~\ref{theorem:main}, we have $K_\psi= 0$ but $K^c_\psi\neq 0$. Finally, Lemma~\ref{lemma:minimal} tells us that $\Phi_c(\iden_d)$ has full rank in $\M{d_{\operatorname{env}}}$, so that 
\begin{equation}
    \operatorname{Tr}(K^c_\psi\Phi_c(\iden_d))>0=\operatorname{Tr}(K_\psi \Phi(\iden_d))\implies \mathcal{Q}^{(1)}(\Phi_c)>0. \qedhere
\end{equation}
\end{proof}
\begin{remark}\label{remark:vikesh}
Theorem~\ref{theorem:main} is more general than Corollary~\ref{corollary:vikesh}, since it can be used to detect positivity of the quantum capacity of channels which lie outside the realm of applicability of Corollary~\ref{corollary:vikesh}, especially for channels $\Phi$ with $d^{*}_{\out}(\Phi)=d^{*}_{\operatorname{env}}(\Phi)$. In Theorem~\ref{theorem:dephasing}, we show that the class of {dephasing} or {Hadamard} channels contains such examples; see also Example~\ref{subsec:equaloutenv}. Moreover, there exist channels $\Phi$ for which even though $d^{*}_{\out}(\Phi)\neq d^{*}_{\env}(\Phi)$, there is no pure state that gets mapped to an output state with maximal rank, so that Corollary~\ref{corollary:vikesh} cannot be applied. The well-known {Werner-Holevo} channel provides such an example. Using Theorem~\ref{theorem:main}, we obtain positivity of the quantum capacity of its complement for all input dimensions $d\ge 4$ in Theorem~\ref{theorem:werner-holevo}.
\end{remark}

Whenever $d^*_{\out}(\Phi)\neq d^*_{\env}(\Phi)$ for some channel $\Phi$, Corollary~\ref{corollary:vikesh} can be applied if the existence of a pure input state that gets mapped to a state with maximal rank in the output space can be guaranteed, see Eq.~\eqref{eq:max-pure-rank}. However, this condition can be difficult to check in practice. In order to see why, let us consider a pair of complementary channels $\Phi:\M{d}\to\M{d_{\out}}$ and $\Phi_c:\M{d}\to \M{d_{\env}}$ with $d^*_{\out}(\Phi)=d_{\out}$ and $d^*_{\env}(\Phi)=d_{\env}$. Then, the associated isometry $V:\C{d}\to\C{d_{\out}}\otimes \C{d_{\env}}$ which defines these channels (see Eq.~\eqref{eq:complementary}) identifies the input space $\C{d}$ with a $d-$dimensional subspace $\operatorname{range}V\subseteq \C{d_{\out}}\otimes \C{d_{\env}}$, which can further be identified with a $d-$dimensional matrix subspace $\operatorname{vec}^{-1}(\operatorname{range}V)\subseteq \M{d_{\out}\times d_{\env}}$ via the inverse of the vectorization map, see Eq.~\eqref{eq:vec}. Since
\begin{equation}
    \forall \ketbra{\psi}\in \St{d}: \quad \operatorname{rank}\Phi(\ketbra{\psi}) = \operatorname{Schmidt\, rank} (V\ket{\psi}) = \operatorname{rank}(\operatorname{vec}^{-1}V\ket{\psi}),
\end{equation}
establishing the existence of a pure state with maximal output rank is equivalent to establishing the existence of a full rank matrix in the subspace $\operatorname{vec}^{-1}(\operatorname{range}V)\subseteq \M{d_{\out}\times d_{\env}}$, which is known to be hard  \cite{Lovsz1989rank,Gurvits2003rank}. This is why the simple dimensional inequalities given below are so convenient. Put simply, they ensure the existence of the desired pure input state whenever $d^*_{\out}(\Phi$) or $d^*_{\env}(\Phi)$ is sufficiently large for an arbitrary channel $\Phi$. We should point out that the following inequalities were also derived in \cite{siddhu2020logsingularities}, albeit in a different way.

\begin{corollary}\label{corollary:vikesh2}
Let $\Phi:\M{d}\rightarrow \M{d_{\out}}$ and $\Phi_c:\M{d}\rightarrow \M{d_{\env}}$ be complementary channels. Then,
\begin{itemize}
    \item $\mathcal{Q}(\Phi)\geq\mathcal{Q}^{(1)}(\Phi)>0$ if $d^{*}_{\out}(\Phi)>d^{*}_{\operatorname{env}}(\Phi)$ and $d^{*}_{\out}(\Phi)>d(d^{*}_{\operatorname{env}}(\Phi)-1)$.
    \item $\mathcal{Q}(\Phi_c)\geq\mathcal{Q}^{(1)}(\Phi_c)>0$ if $d^{*}_{\env}(\Phi)>d^{*}_{\out}(\Phi)$ and $d^{*}_{\env}(\Phi)>d(d^{*}_{\out}(\Phi)-1)$.
\end{itemize}
\end{corollary}
\begin{proof}
Let us begin with the second part. Assume that the channels are minimally defined; see Remark~\ref{remark:iso-cap}. In view of Corollary~\ref{corollary:vikesh}, our goal is to show that if $d_{\env}>d_{\out}$ and $d_{\operatorname{env}}>d(d_{\out}-1)$, then there exists $\ketbra{\psi}{\psi}\in \St{d}$ such that $\operatorname{rank}\Phi(\ketbra{\psi}{\psi})=d_{\out}$. Suppose that this is not the case. Then, for every $\ket{\psi}\in \C{d}$, $\Phi(\ketbra{\psi})\in \M{d_{\out}}$ must be singular $\implies \ker\Phi(\ketbra{\psi}{\psi})\neq\{0\}\implies$ there exists $0\neq\ket{\phi}\in \ker\Phi(\ketbra{\psi}{\psi})$ such that
\begin{align}
    0 = \operatorname{Tr}(\Phi_\psi \ketbra{\phi}{\phi} ) = \operatorname{Tr}\left[ \left(\Phi_\psi\otimes\ketbra{\Bar{\phi}}{\Bar{\phi}}\right) \ketbra{\Omega}{\Omega} \right] 
    = \operatorname{Tr} \left[\ketbra{\psi\Bar{\phi}}{\psi\Bar{\phi}} J(\Phi^*) \right],
\end{align}
where $\Phi_\psi = \Phi(\ketbra{\psi}{\psi})$, $\ket{\Omega}\in \C{d_{\out}}\otimes\C{d_{\out}}$ is defined as in Eq.~\eqref{eq:choi}, and $J(\Phi^*)= (\Phi^*\otimes \operatorname{id})\ketbra{\Omega}{\Omega}$. Hence, for every $\ket{\psi}\in\C{d}$, there exists $\ket{\phi}\in \C{d_{\out}}$ such that $\ket{\psi\phi}\in \ker J(\Phi^*)$. In particular, there are at least $d$ linearly independent vectors in $\ker J(\Phi^*)$, so that according to the rank-nullity theorem, $d_{\operatorname{env}}=\operatorname{rank}J(\Phi)=\operatorname{rank}J(\Phi^*)\leq d(d_{\out}-1)$, which contradicts our original assumption. An identical argument can be used to prove the first part as well.
\end{proof}

By exploiting an intriguing matrix-theoretic result (Lemma~\ref{lemma:dim-rank-bound}), we next derive an interesting dimensional inequality in Corollary~\ref{corollary-biginput}, which ensure the existence of the sought-after pure state (i.e. the state which gets mapped onto an output state with maximal rank) for any given channel. Unlike Corollary~\ref{corollary:vikesh2}, the inequality in Corollary~\ref{corollary-biginput} implies that all channels with sufficiently large {input} dimension must either have positive quantum capacity or positive complementary quantum capacity. To see how these two corollaries nicely complement each other, see Remark~\ref{remark:dim-ineq}.

\begin{lemma}\label{lemma:dim-rank-bound}
Let $\mathcal{S}\subseteq \M{d_1\times d_2}$ (with $d_1\leq d_2$) be a matrix subspace such that $\displaystyle\max_{X\in\mathcal{S}}\operatorname{rank}X = r$. Then, $\operatorname{dim}\mathcal{S}\leq rd_2$.
\end{lemma}
\begin{proof}
See the Methods section~\ref{appen:lemmadim-rank-bound}.
\end{proof}

\begin{corollary}\label{corollary-biginput}
Let $\Phi:\M{d}\rightarrow \M{d_{\out}}$ and $\Phi_c:\M{d}\to \M{d_{\env}}$ be complementary channels. Then,
\begin{itemize}
    \item $\mathcal{Q}(\Phi)\geq\mathcal{Q}^{(1)}(\Phi)>0$ if $d^{*}_{\out}(\Phi)>d^{*}_{\operatorname{env}}(\Phi)$ and $d>d^{*}_{\out}(\Phi)(d^{*}_{\operatorname{env}}(\Phi)-1)$.
    \item $\mathcal{Q}(\Phi_c)\geq\mathcal{Q}^{(1)}(\Phi_c)>0$ if $d^{*}_{\operatorname{env}}(\Phi)>d^{*}_{\out}(\Phi)$ and $d>d^{*}_{\operatorname{env}}(\Phi)(d^{*}_{\out}(\Phi)-1)$.
\end{itemize}
\end{corollary}
\begin{proof}
We only prove the first part here, and leave an analogous proof of the second part to the reader. As usual, we assume that the channels $\Phi:\M{d}\rightarrow \M{d_{\out}}$ and $\Phi_c:\M{d}\rightarrow \M{d_{\operatorname{env}}}$ are minimally defined, so that $d_{\out} = d^{*}_{\out}(\Phi)$ and $d_{\env} = d^{*}_{\operatorname{env}}(\Phi)$ (see Remark~\ref{remark:iso-cap}). Then, if $V:\C{d}\to \C{d_{\out}}\otimes \C{d_{\env}}$ is the associated Stinespring isometry, the matrix subspace $\mathcal{S}=\operatorname{vec}^{-1}(\operatorname{range}V)\subseteq \M{d_{\out}\times d_{\env}}$ has $\dim \mathcal{S}=d>d_{\out}(d_{\env}-1)$. Since $d_{\env}<d_{\out}$, Lemma~\ref{lemma:dim-rank-bound} tells us that there exists $X\in \mathcal{S}$ with $\operatorname{rank}X=d_{\env}$, so that for $\ket{\psi}\in \C{d}$ such that $V\ket{\psi}=\operatorname{vec}X$, we have  
\begin{equation}
    \operatorname{rank}\Phi(\ketbra{\psi}) = \operatorname{Schmidt\,rank}( V\ket{\psi}) = \operatorname{rank}X = d_{\env}.
\end{equation}
Hence, Corollary~\ref{corollary:vikesh} can be readily applied to obtain the desired result. 
\end{proof}

\begin{remark} \label{remark:dim-ineq}
Let us fix two positive integers $d_{\out}>d_{\env}$ and consider the set of all channels $\Phi$ with $d^*_{\out}(\Phi) = d_{\out}$ and $d^*_{\env}(\Phi)=d_{\env}$. Then, the input dimension of these channels can be any positive integer $d \leq d^*_{\out}(\Phi)d^*_{\env}(\Phi)$. Barring the $d=1$ case, Corollaries~\ref{corollary:vikesh2} and \ref{corollary-biginput} show that whenever $d$ lies within the two opposite extremes of this range: $d<d^*_{\out}(\Phi)/(d^*_{\env}(\Phi)-1)$ or $d>d^*_{\out}(\Phi)(d^*_{\env}(\Phi)-1)$, the corresponding channel $\Phi$ has positive quantum capacity.
\end{remark} 

By exploiting Theorem~\ref{theorem:main}, we now introduce a powerful necessary condition for a quantum channel to have zero capacity.

\begin{theorem}\label{theorem:zero-necessary}
Let $\Phi:\M{d}\rightarrow \M{d_{\out}}$ and $\Phi_c:\M{d}\to\M{d_{\env}}$ be complementary channels. For a pure input state $\ketbra{\psi}\in\St{d^{\otimes n}}$, denote the orthogonal projections onto $\operatorname{range}\Phi^{\otimes n}(\ketbra{\psi}{\psi})$ and $\operatorname{range}\Phi^{\otimes n}_c(\ketbra{\psi}{\psi})$ by $R_\psi$ and $R^c_\psi$, respectively. Then, if $\mathcal{Q}(\Phi)=0$, the following relation holds
\begin{equation}
    \forall n\in\mathbb{N}, \forall \ketbra{\psi}{\psi}\in \St{d^{\otimes n}}: \quad (\Phi_c^{\otimes n})^*(R^c_\psi)\leq (\Phi^{\otimes n})^*(R_\psi).
\end{equation}
\begin{proof}
Since $\mathcal{Q}(\Phi) = 0 \implies \mathcal{Q}^{(1)}(\Phi^{\otimes n})=0$ for all $n\in\mathbb{N}$, we can apply Theorem~\ref{theorem:main} to $\Phi^{\otimes n}$ for each $n$ to obtain the required implication. Let us spell out the details for the $n=1$ case. Here, Theorem~\ref{theorem:main} forces $\operatorname{Tr}(K_\psi \Phi(\sigma)) \leq \operatorname{Tr}(K^c_\psi \Phi_c(\sigma))$ for all $\ketbra{\psi}{\psi}, \sigma\in \St{d}$, where $K_\psi$ and $K^c_\psi$ are as defined in Theorem~\ref{theorem:main}. Hence, the following equivalences hold for each pure state $\ketbra{\psi}{\psi}\in \St{d}$:
\begin{equation}
\forall\sigma\in \St{d}: \quad \operatorname{Tr}[(\Phi^*_c(K^c_\psi)-\Phi^*(K_\psi))\sigma] \geq 0 \iff \Phi^*_c(K^c_\psi) \geq \Phi^*(K_\psi)\iff \Phi^*_c(R^c_\psi)\leq \Phi^*(R_\psi),    
\end{equation}
where the equivalence of the latter two operator inequalities in $\M{d}$ is a consequence of the fact that both $\Phi^*$ and $\Phi^*_c$ are unital, $R_\psi+K_\psi=\iden_{d_{\out}}$ and $R^c_\psi+K^c_\psi=\iden_{d_{\operatorname{env}}}$.
\end{proof}
\end{theorem}

As a sanity check, the next lemma shows that anti-degradable and transpose anti-degradable channels (which are known to have zero capacity, see Section~\ref{subsec:def}) satisfy the necessary condition stated in Theorem~\ref{theorem:zero-necessary}.

\begin{lemma}\label{lemma:anti-degaradable}
Let $\Phi:\M{d}\rightarrow \M{d_{\out}}$ be an anti-degradable or transpose anti-degradable channel and $\Phi_c:\M{d}\to\M{d_{\env}}$ be complementary to $\Phi$. Then,
\begin{equation}
    \forall n\in\mathbb{N}, \forall \ketbra{\psi}{\psi}\in \St{d^{\otimes n}}: \quad (\Phi_c^{\otimes n})^*(R^c_\psi)\leq (\Phi^{\otimes n})^*(R_\psi),
\end{equation}
where $R_\psi$ and $R^c_\psi$ are as defined in Corollary~\ref{theorem:zero-necessary}.
\end{lemma}
\begin{proof}
We prove the result for anti-degradable channels and leave an almost identical proof of the transposed version to the reader. It suffices to obtain the result for $n=1$, since if $\Phi$ is anti-degradable, then $\Phi^{\otimes n}$ is also anti-degradable for all $n\in\mathbb{N}$. To begin with, note that the anti-degradability of $\Phi:\M{d}\rightarrow\M{d_{\out}}$ implies that there exists a channel $\mathcal{N}:\M{d_{\operatorname{env}}}\rightarrow\M{d_{\out}}$ such that $\Phi=\mathcal{N}\circ \Phi_c$ and $\Phi^*=\Phi^*_c\circ \mathcal{N}^*$. Hence, for every pure state $\ketbra{\psi}{\psi}\in \St{d}$, we can exploit Lemma~\ref{lemma:3} to obtain the following sequence of implications: 
\begin{align}
    \operatorname{range}\mathcal{N}(R^c_\psi)=\operatorname{range}R_\psi \implies \mathcal{N}(R^c_\psi)K_\psi = 0 &\implies \operatorname{Tr}(\mathcal{N}(R^c_\psi)K_\psi) = 0 = \operatorname{Tr}(R^c_\psi \mathcal{N}^*(K_\psi)) \nonumber \\  
    &\implies R^c_\psi\mathcal{N}^*(K_\psi)=0 \nonumber \\
    &\implies \operatorname{range}\mathcal{N}^*(K_\psi) \subseteq \ker R^c_\psi = \operatorname{range} K^c_\psi. \label{eq:anti}
\end{align}
Observe that while obtaining the implications above, we used the fact that for $A,B\geq 0$, $AB=0\iff \operatorname{Tr}(AB)=0$. Now, since $\mathcal{N}^*$ is unital and completely positive (being the adjoint of a quantum channel), it is contractive in the operator norm (see \cite[Theorem 2.3.7]{bhatia2015positive}), which implies that $\Vert \mathcal{N}^*(K_\psi)\Vert \leq \Vert K_\psi \Vert \leq 1$. Combining Eq.~\eqref{eq:anti} with the previous result, we obtain $\mathcal{N}^*(K_\psi)\leq K^c_\psi \implies \Phi^*(K_\psi) = \Phi^*_c(\mathcal{N}^*(K_\psi)) \leq \Phi^*_c(K^c_\psi) \implies \Phi_c^*(R^c_\psi)\leq \Phi^*(R_\psi) $.
\end{proof}

\begin{remark}
It would be desirable to obtain a result analogous to Lemma~\ref{lemma:anti-degaradable} for PPT channels, which also have zero quantum capacity. Observe that Lemma~\ref{lemma:anti-degaradable} already contains the desired result for transpose anti-degradable channels, which form a subclass of PPT channels. We should emphasize, however, that the relationship between a (transpose) anti-degradable channel and its complement -- which is crucially exploited in the proof of Lemma~\ref{lemma:anti-degaradable} -- does not hold for PPT channels in general. Hence, the aforementioned proof would not work in this case. 
\end{remark}

In the following sections, we apply our main results to detect positive quantum capacities of several different families of quantum channels. We will also give a miscellaneous application of our main result by extending certain existing structure theorems for the class of degradable quantum channels to the strictly larger class of more capable quantum channels.

\section{Applications}

\subsection{Depolarizing and transpose depolarizing channels}\label{subsec:depol}

The {depolarizing channels} $\mathcal{D}_p:\M{d}\to \M{d}$ and the {transpose depolarizing channels} $\mathcal{D}^\top_q:\M{d}\to \M{d}$ form one-parameter families of completely positive and trace-preserving linear maps within their respective parameter ranges $p\in [0,\frac{d^2}{d^2-1}]$ and $q\in [\frac{d}{d+1}, \frac{d}{d-1}]$, and are defined as follows:
\begin{equation}
    \forall X\in \M{d}: \quad \mathcal{D}_p(X) = (1-p)X + p\operatorname{Tr}(X) \frac{\iden_d}{d} \quad\text{and} \quad \mathcal{D}^\top_q(X) = (1-q)X^\top + q\operatorname{Tr}(X) \frac{\iden_d}{d}.
\end{equation}

These families of channels are arguably two of the most prominent noise models in quantum information theory. While considerable effort has gone into computing the quantum capacity of depolarizing channels \cite{Smith2008depolarizing, Sutter2017depolarizing, Leditzky2018depolarizing}, an analysis of their complementary quantum capacities has largely evaded the spotlight. Recently, by carefully scrutinizing the coherent information of the {qubit} depolarizing channels $\mathcal{D}_p:\M{2}\to \M{2}$, Watrous and Leung have shown that the these channels have positive complementary quantum capacities for all $p>0$ \cite[Theorem 1]{Leung2017complementary}. A similar analysis has also been performed for the {qubit} $\mathcal{D}^\top_q:\M{2}\to \M{2}$ and {qutrit} $\mathcal{D}^\top_q:\M{3}\to \M{3}$ transpose-depolarizing channels to establish positivity of their complementary quantum capacities for certain values of the parameter $q$, see \cite[Section 3.1]{Bradler2015depolarizing}. We now substantially generalize the above results to obtain positivity of the complementary quantum capacities of depolarizing and transpose-depolarizing channels in arbitrary dimensions while also providing much simpler proofs. 

\begin{theorem}\label{theorem:depol-transp}
Let $\mathcal{D}_p:\M{d}\to \M{d}$ and $\mathcal{D}^\top_q:\M{d}\to \M{d}$ be the qudit depolarizing and transpose-depolarizing channels with $p\in [0,\frac{d^2}{d^2-1}]$ and $q\in [\frac{d}{d+1},\frac{d}{d-1}]$. Then,
\begin{itemize}
    \item if $p>0$, any complementary channel $\mathcal{D}_p^c\in \mathcal{C}_{\mathcal{D}_p}$ has positive quantum capacity.
    \item if $q<\frac{d}{d-1}$, any complementary channel $(\mathcal{D}^\top_p)^c\in \mathcal{C}_{\mathcal{D}^\top_p}$ has positive quantum capacity. 
\end{itemize}

\end{theorem}
\begin{proof}
Firstly, observe that for all allowed $p,q$, both $\mathcal{D}_p$ and $\mathcal{D}^\top_q$ are unital so that $d^*_{\out}(\mathcal{D}_p) = d^*_{\out}(\mathcal{D}^\top_q) =d$, see Remark~\ref{remark:minimal}. Moreover, the Choi matrices of these channels can be written as
\begin{equation}
J(\mathcal{D}_p) = \frac{p}{d} (\iden_d \otimes \iden_d) + (1-p) \sum_{i,j=0}^{d-1} \ketbra{ii}{jj} \quad\text{and}\quad J(\mathcal{D}^\top_q) = \frac{q}{d} (\iden_d \otimes \iden_d) + (1-q) \sum_{i,j=0}^{d-1} \ketbra{ij}{ji}.  
\end{equation}
These are, respectively, the well-known Isotropic \cite{Horodecki1999iso} and Werner \cite{Werner1989} matrices, which admit nice block diagonal decompositions so that their ranks can be easily computed (see \cite[Example 3.3 and Corollary 4.2]{singh2020diagonal} or Lemma~\ref{lemma:(C)DUC-min}):
\begin{align}
    d^*_{\env}(\mathcal{D}_p) = \operatorname{rank} J(\mathcal{D}_p) &=
    \begin{cases}
    1 &\quad \text{if }\,\, p=0 \\
    d^2 &\quad \text{if }\,\, 0<p<\frac{d^2}{d^2-1} \\
    \geq d^2-d+1 &\quad \text{if }\,\, p=\frac{d^2}{d^2-1}
    \end{cases} \\
    d^*_{\env}(\mathcal{D}^\top_q) = \operatorname{rank} J(\mathcal{D}^\top_q) &=
    \begin{cases}
    d(d+1)/2 &\quad\quad \text{if }\,\, q=\frac{d}{d+1} \\
    d^2 &\quad\quad \text{if }\,\, \frac{d}{d+1}<q<\frac{d}{d-1} \\
    d(d-1)/2 &\quad\quad \text{if }\,\, q=\frac{d}{d-1}
    \end{cases}
\end{align}
Now, since for all $p>0$ and $\frac{d}{d+1}<q<\frac{d}{d-1}$, we have $d^{*}_{\env}(\Phi)> d^{*}_{\out}(\Phi)$ and $d^{*}_{\env}(\Phi)>d[d^{*}_{\out}(\Phi)-1]$ for $\Phi \in \{\mathcal{D}_p, \mathcal{D}^\top_q\}$, the desired results hold because of Corollary~\ref{corollary:vikesh2}. For $q=\frac{d}{d+1}$, it is easy to see that $\mathcal{D}^\top_q$ maps any pure state to a full rank state on the output, so that Corollary~\ref{corollary:vikesh} can be applied to obtain positivity of the complementary quantum capacity.
\end{proof}

\begin{remark}
The depolarizing- and transpose-depolarizing channels have the following covariance property for all $U\in\mathcal{U}_d$ and $X\in\M{d}$:
\begin{equation}
 \mathcal{D}_p(UXU^\dagger) = U\mathcal{D}_p(X)U^\dagger \quad\text{and}\quad \mathcal{D}^\top_q(UXU^\dagger) = \overbar{U} \mathcal{D}^\top_q(X)U^\top,
\end{equation}
where $\mathcal{U}_d$ is the unitary group in $\M{d}$. Upon relaxing the above covariance condition to hold only for the smaller group of diagonal unitary matrices $\mathcal{DU}_d\subset \mathcal{U}_d$, we obtain substantially bigger classes of channels which are defined by $\sim d^2$ real parameters (as opposed to the single real parameters $p$ and $q$ in the depolarizing and transpose-depolarizing families). Theorem~\ref{theorem:CDUC-cap} serves as an analogue of Theorem~\ref{theorem:depol-transp} for these larger classes of channels.
\end{remark}

Let us now analyse the depolarizing- and transpose-depolarizing channels present at the extreme ends of their respective parameter ranges. For $p=0$, $\mathcal{D}_p=\operatorname{id}:\M{d}\to\M{d}$ is the identity channel and $d^{*}_{\env}(\mathcal{D}_p) = 1$ for all $d\in\mathbb{N}$, so that any complementary channel has zero quantum capacity.  The $q=\frac{d}{d-1}$ case of the transpose-depolarizing channel $\mathcal{D}^\top_q$ is much more subtle. For $d=2$, the channel $\mathcal{D}^\top_2:\M{2}\to \M{2}$ has $d^*_{\env}(\mathcal{D}^\top_2)=1$ so that any complementary channel has zero capacity. When $d=3$, it is known that $\mathcal{D}^\top_{3/2}:\M{3}\to \M{3}$ is both degradable and anti-degradable  \cite{Cubitt2008degradable} (in fact, there exists a complementary channel $[\mathcal{D}^\top_{3/2}]^c$ such that $[\mathcal{D}^\top_{3/2}]^c=\mathcal{D}^\top_{3/2}$) and hence any complementary channel again has zero capacity. For all $d\geq 4$, Theorem~\ref{theorem:werner-holevo} below settles the question. It is perhaps worthwhile to point out that the transpose-depolarizing channel at the extreme parameter value $q=\frac{d}{d-1}$ is more widely known as the {Werner-Holevo} channel \cite{Werner-Holevo}.

\begin{equation}\label{eq:werner-holevo}
    \forall X\in \M{d}: \quad  \Phi_{\rm{WH}}(X) \equiv \mathcal{D}^\top_{\frac{d}{d-1}}(X) = \frac{\operatorname{Tr}(X)\iden_d-X^\top}{d-1}.
\end{equation}

\begin{theorem}
\label{theorem:werner-holevo}
Any complementary channel to the Werner-Holevo channel $\Phi_{\rm{WH}}:\M{d}\to \M{d}$ has positive quantum capacity when $d\geq 4$.
\end{theorem}
\begin{proof}
Let us denote the Werner-Holevo channel simply by $\Phi:\M{d}\to \M{d}$. We already know that $d(d-1)/2 = d^*_{\env}(\Phi) > d^*_{\out}(\Phi)=d$ whenever $d\geq 4$. However, Corollary~\ref{corollary:vikesh} cannot be applied to this channel, since it maps every pure state to a state with rank $=d-1$, see Eq.~\eqref{eq:werner-holevo}. Thus, our aim is to exploit Theorem~\ref{theorem:main} instead. Let us consider the pure state $\ketbra{0}{0}\in \St{d}$, (recall that $\{\ket{k}\}_{k=0}^{d-1}$ denotes the standard basis of ${\mathbb{C}}^d$), so that the orthogonal projection onto $\ker \Phi(\ketbra{0}{0})$ becomes $K_0=\ketbra{0}{0}$. Now, it was shown in \cite{Holevo2007complementary} that $\Phi$ admits a complementary channel $\Phi_c:\M{d}\to \M{d}\otimes \M{d}$ of the following form:
\begin{equation}
    \forall X\in\M{d}: \qquad \Phi_c(X) = \frac{1}{2(d-1)}(\iden_d\otimes \iden_d - F)(X\otimes \iden_d)(\iden_d\otimes \iden_d - F),
\end{equation}
where $F:\C{d}\otimes\C{d}\to \C{d}\otimes\C{d}$ is the {flip} operator whose action on the canonical basis $\{\ket{i}\otimes\ket{j}\}_{i,j=0}^{d-1}$ of $\C{d}\otimes\C{d}$ is defined as $F\ket{i}\otimes\ket{j}=\ket{j}\otimes\ket{i}$. Let us now see how $\Phi_c$ acts on $\ketbra{0}$:
\begin{align}
    \Phi_c(\ketbra{0}{0}) &= \frac{1}{2(d-1)}[\ketbra{0}\otimes\iden_d - (\ketbra{0}\otimes\iden_d) F - F (\ketbra{0}\otimes\iden_d) + F(\ketbra{0}\otimes\iden_d)F ] \nonumber \\
    &= \frac{1}{2(d-1)}\sum_{i=0}^{d-1}[ \ketbra{0i}{0i} - \ketbra{0i}{i0} - \ketbra{i0}{0i} + \ketbra{i0}{i0}] \nonumber \\
    &= \frac{1}{2(d-1)}\sum_{i=0}^{d-1} (\ket{0i}-\ket{i0}) (\bra{0i}-\bra{i0}) \label{eq:kc_0}.
\end{align}
Clearly, $\operatorname{range}\Phi_c(\ketbra{0})$ lies within the anti-symmetric subspace of $\C{d}\otimes\C{d}$. If
\begin{equation}
    \ket{\psi^+_{ij}} = \begin{cases}
     (\ket{ij} + \ket{ji})/\sqrt{2}, &\text{if } i < j\\
     \,\,\ket{ii}, &\text{if }i=j
\end{cases} \qquad\text{and}\qquad \ket{\psi^-_{ij}} = \frac{\ket{ij}-\ket{ji}}{\sqrt{2}} \quad\text{for } i<j,
\end{equation}
then the orthogonal projections onto the symmetric and anti-symmetric subspaces of $\C{d}\otimes\C{d}$ can be written as $P_s=\sum_{i\leq j}\ketbra{\psi^+_{ij}}$ and $P_a=\sum_{i<j}\ketbra{\psi^-_{ij}}$, respectively. Using Eq.~\eqref{eq:kc_0}, we can express the orthogonal projection onto $\ker \Phi_c(\ketbra{0})$ as follows
\begin{equation}
    K^c_0 = P_s + \sum_{i\neq 0; i<j}\ketbra{\psi^-_{ij}}.
\end{equation}
Finally, since $\Phi(\iden_d)=\iden_d$ and $\Phi_c(\iden_d)=\frac{1}{d-1}(\iden_d\otimes\iden_d - F)$, the required trace expressions can be easily calculated. We have $\operatorname{Tr}(K_0\Phi(\iden_d))=1$ and for all $d\geq 4$,
\begin{align}
    \operatorname{Tr}(K^c_0\Phi_c(\iden_d)) &= \frac{1}{d-1} \operatorname{Tr}\left[\left(P_s + \sum_{i\neq 0; i<j}\ketbra{\psi^-_{ij}}\right)(\iden_d\otimes\iden_d - P_s + P_a)\right] \nonumber\\
    &= \frac{2}{d-1} \operatorname{Tr}\left[ \sum_{i\neq 0; i<j}\ketbra{\psi^-_{ij}} \right] = \frac{2}{d-1}\left[\frac{d(d-1)}{2}-(d-1)\right]=d-2 > 1.
\end{align}
A direct application of Theorem~\ref{theorem:main} suffices to conclude that $\mathcal{Q}^{(1)}(\Phi_c)>0$.
\end{proof}

\subsection{Generalized Pauli  (or Weyl Covariant) Channels}
For an entrywise non-negative matrix $P\in \M{d}$ satisfying $\sum_{ij}p_{ij}=1$, the {generalized Pauli channel} $\Phi_P:\M{d}\to \M{d}$ associated with $P$ is a mixed unitary channel defined as follows:
\begin{equation}\label{eq:pauli}
    \forall X\in \M{d}: \qquad \Phi_P(X) = \sum_{0\leq i,j \leq d-1} p_{ij} U_{ij}X U_{ij}^\dagger,
\end{equation}
where $\{U_{ij}\}_{i,j=0}^{d-1}$ is the unitary orthogonal basis of $\M{d}$ formed by the discrete Heisenberg–Weyl operators:
\begin{equation}
    U_{ij} = X^i Z^j, \quad \text{where}\quad X = \sum_{k=0}^{d-1} \ketbra{k+1}{k} \quad\text{and}\quad Z = \sum_{k=0}^{d-1} \omega^k \ketbra{k}{k}
\end{equation}
are the so-called {shift} and {clock} matrices, respectively. Note that in the above definition, addition inside kets happens modulo $d$ and $\omega = e^{2\pi i/d}$ is the $d^{th}$ root of unity. The orthogonality of the unitary basis $\{U_{ij}\}_{i,j=0}^{d-1}$ can be easily verified by showing that $\operatorname{Tr}(U^\dagger_{ij}U_{kl}) = d\delta_{ik}\delta_{jl}$. When $d=2$, these unitary matrices are identical to the familiar $2\times 2$ pauli matrices. In this case, a sophisticated analysis of the coherent information of the qubit Pauli channels reveals that these channels have positive complementary quantum capacity whenever the associated matrix $P\in \M{2}$ has at least $3$ non-zero entries \cite[Theorem 2]{Leung2017complementary}. Our next theorem extends this result for generalized Pauli channels acting in arbitrary dimensions. 

\begin{theorem}\label{theorem:pauli}
Let $\Phi_P:\M{d}\to \M{d}$ be a generalized Pauli channel associated with $P\in\M{d}$. If $P$ has at least $d+1$ non-zeros entries distributed in such a way that either every row or every column of $P$ is non-zero, then any complementary channel $\Phi^c_P\in \mathcal{C}_{\Phi_P}$ has positive quantum capacity.
\end{theorem}
\begin{proof}
Firstly, since $\Phi_P$ is unital, we have $d^*_{\out}(\Phi_P)=d$. In addition, if $P$ has at least $d+1$ non-zero entries, then $d^*_{\env}(\Phi_P)\geq d+1$ (see Remark~\ref{remark:minimal}), so that Corollary~\ref{corollary:vikesh} can be applied to obtain the desired result if there exists a pure state $\ketbra{\psi}{\psi}\in \St{d}$ such that $\Phi_P(\ketbra{\psi})\in\St{d}$ has full rank. Now, if $P$ has a non-zero entry $p_{ik_i}\neq 0$ for each $i\in \{0,1,\ldots ,d-1\}$, then every power of the shift matrix $X$ is present in the Kraus decomposition of $\Phi_P$ (see Eq.~\eqref{eq:pauli}). Hence, by choosing $\ket{\psi}$ to be an eigenvector of the clock matrix $Z$ (say $\ket{\psi}=\ket{0}$), we get a full rank state on the output: $\Phi_P(\ketbra{0})\geq \sum_{i=0}^{d-1}p_{ik_i}\ketbra{i}$. Similarly, if every column of $P$ has a non-zero entry, then an eigenvector of the shift matrix $X$, say $\ket{\psi}=\frac{1}{\sqrt{d}}\sum_{i=0}^{d-1}\ket{i}$, would result in a full rank output. 
\end{proof}

\subsection{Multi-level amplitude damping channels}
{Multi-level amplitude damping} (MAD) channels are employed to model the decay dynamics of a particle in a $d$-level quantum system with the associated Hilbert space $\mathbb{C}^d = \operatorname{span}\{\ket{i}\}_{i=0}^{d-1}$. A total of $d(d-1)/2$ real numbers $\{\gamma_{j\to i}\}_{0\leq i < j\leq d-1}$ are used to parameterize the decay rates of the particle from a higher $j^{th}$ level to a lower $i^{th}$ level, so that the MAD channel $\Phi_{\Vec{\gamma}}:\M{d}\to \M{d}$ admits the following Kraus representation
\begin{equation} \label{eq:MAD}
    \forall X\in \M{d}: \qquad \Phi_{\Vec{\gamma}} (X) = A_0 X A_0^\dagger + \sum_{0\leq i<j\leq d-1} A_{ij} X A_{ij}^\dagger,
\end{equation}
where $\Vec{\gamma}\in \mathbb{R}^{d(d-1)/2}$, $A_0 = \ketbra{0}{0} + \sum_{j=1}^{d-1} \sqrt{1-\sum_{i<j} \gamma_{j\to i}} \ketbra{j}{j}$ and $A_{ij} = \sqrt{\gamma_{j\to i}}\ketbra{i}{j}$ for all $i<j$. Complete positivity and trace preserving property of $\Phi_{\Vec{\gamma}}$ implies that
\begin{equation}
    \forall i<j: \quad\gamma_{j\to i} \geq 0 \qquad\text{and}\qquad \forall j\in[d]: \quad \sum_{i<j}\gamma_{j\to i} \leq 1.
\end{equation}

We say that a level $j\in \{1,\ldots , d-1\}$ gets {totally depleted} under the channel $\Phi_{\Vec{\gamma}}$ if for all $k>j$, $\gamma_{k\to j}=0$ (i.e., no higher level has a positive rate of decaying to the $j^{th}$ level) and $\sum_{i<j}\gamma_{j\to i}=1$ (i.e., the sum of decay rates from the $j^{th}$ level to all the lower levels equals one). Since there are no levels below the ground level $j=0$, it can never get totally depleted. Given $\Phi_{\Vec{\gamma}}$, we denote the number of levels which do not get totally depleted (including the ground level) by $n_1(\Vec{\gamma})$. We denote the total number of non-zero decay rates in $\Vec{\gamma}$ by $n_2 (\Vec{\gamma})$. The following lemma relates these newly defined parameters with the minimal output and environment dimensions of $\Phi_{\Vec{\gamma}}$.

\begin{lemma}\label{lem:n1n2}
Let $\Phi_{\Vec{\gamma}}:\M{d}\to \M{d}$ be a MAD channel parameterized by $\Vec{\gamma}\in \mathbb{R}^{d(d-1)/2}$. Then, 
\begin{equation}
    d^{*}_{\out}(\Phi_{\Vec{\gamma}}) = n_1 (\Vec{\gamma}) \qquad\text{and}\qquad d^{*}_{\env}(\Phi_{\Vec{\gamma}}) = 1+n_2 (\Vec{\gamma}).
\end{equation}
\end{lemma}
\begin{proof}
Apply the channel $\Phi_{\Vec{\gamma}}$ on the identity matrix $\iden_d\in\M{d}$ to deduce that $\operatorname{rank}\Phi_{\Vec{\gamma}}(\iden_d)=d^*_{\out}(\Phi_{\Vec{\gamma}})=n_1(\Vec{\gamma})$. To obtain the expression for $d^*_{\env}(\Phi_{\Vec{\gamma}})$, observe that the number of linearly independent Kraus operators in the Kraus representation of $\Phi_{\Vec{\gamma}}$ (see Eq.~\eqref{eq:MAD}) is $1+n_2(\Vec{\gamma})$, so that $\operatorname{rank}J(\Phi_{\Vec{\gamma}})=d^*_{\env}(\Phi_{\Vec{\gamma}})=1+n_2(\Vec{\gamma})$, see Remark~\ref{remark:minimal}.
\end{proof}

The quantum capacity of the $2$-level (or qubit) amplitude damping channels is quite well understood \cite{Vittorio2005anti}. This is because any such channel is completely determined by a single decay rate $\gamma\in\mathbb{R}$, so that $d^{*}_{\env}(\Phi_\gamma)\leq 2$ and these types of channels are known to be either degradable or anti-degradable \cite{Wolf2007capacity}. Unfortunately, the same cannot be said for MAD channels in higher dimensions. Hence, their capacity analysis has recently started attracting significant interest from the community, see \cite{Chessa2021Mad} and references therein. However, to the best of our knowledge, virtually nothing is known about the complementary quantum capacities of these channels. In the following theorem, we obtain the first results in this direction by deriving simple sufficient conditions on the decay rate vector $\Vec{\gamma}\in \mathbb{R}^{d(d-1)/2}$ which ensure that the associated MAD channel $\Phi_{\Vec{\gamma}}$ has positive (complementary) capacity. 

\begin{theorem}\label{theorem:MAD}
Let $\Phi_{\Vec{\gamma}}:\M{d}\to \M{d}$ be a MAD channel parameterized by $\Vec{\gamma}\in \mathbb{R}^{d(d-1)/2}$. 
\begin{itemize}
    \item If $1+n_2(\Vec{\gamma})>n_1(\Vec{\gamma})$ and for every $j\in\{0,1,\ldots ,d-2\}$ which does not get totally depleted, there exists a non-zero decay rate $\gamma_{k\to j}$ for some $k>j$, then any complementary channel $\Phi^c_{\Vec{\gamma}}\in \mathcal{C}_{\Phi_{\Vec{\gamma}}}$ has positive quantum capacity.
    \item If $1+n_2(\Vec{\gamma})<n_1(\Vec{\gamma})$ and every $j\in\{0,1,\ldots ,d-2\}$ which does not get totally depleted  receives a decay contribution from not more than one higher level, then the channel $\Phi_{\Vec{\gamma}}$ has positive quantum capacity.
\end{itemize}
\end{theorem}
\begin{proof}
Let $\ket{e}=\sum_{i=0}^{d-1}\ket{i}$ and $\ket{\Gamma}= \ket{0} + \sum_{j=1}^{d-1} \sqrt{1-\sum_{i<j} \gamma_{j\to i}} \ket{j}$, so that 
\begin{equation}
    \Phi_{\Vec{\gamma}} (\ketbra{e})=\sum_{j}\sum_{k>j} \gamma_{k\to j} \ketbra{j} + \ketbra{\Gamma},
\end{equation}
and the conditions stated in the theorem ensure that
\begin{equation}
    \operatorname{rank}\Phi_{\Vec{\gamma}} (\ketbra{e}) = \begin{cases}
    n_1(\Vec{\gamma}) &\text{ if } 1+n_2(\Vec{\gamma})>n_1(\Vec{\gamma}) \\
    1+n_2(\Vec{\gamma}) &\text{ if } 1+n_2(\Vec{\gamma})<n_1(\Vec{\gamma})
    \end{cases}.
\end{equation}
Then Lemma~\ref{lem:n1n2} and Corollary~\ref{corollary:vikesh} can be readily applied to obtain the required results.
\end{proof}

\begin{remark}
Theorem~\ref{theorem:MAD} can be considered as a special case of a more general result which holds for the class of (conjugate) diagonal unitary covariant channels, see Proposition~\ref{prop:CDUC-rank} and Theorem~\ref{theorem:CDUC-cap}. 
\end{remark}

\subsection{(Conjugate) diagonal unitary covariant channels}

Let us denote the group of diagonal unitary matrices in $\M{d}$ by $\mathcal{DU}_d$. A channel $\Phi:\M{d}\to \M{d}$ is said to be diagonal unitary covariant (DUC) (resp. conjugate diagonal unitary covariant (CDUC)) if for all $U\in\mathcal{DU}_d$ and $X\in\M{d}$,
\begin{equation}
    \Phi(UXU^\dagger) = U^\dagger\Phi(X)U \quad\text{resp.}\quad \Phi(UXU^\dagger) = U\Phi(X)U^\dagger.
\end{equation}
It is easy to see that if $\Phi$ is DUC (resp. CDUC), then for all $U\in\mathcal{DU}_d$,
\begin{equation}
    (U\otimes U)J(\Phi)(U^\dagger\otimes U^\dagger) = J(\Phi) \quad\text{resp.}\quad (U\otimes U^\dagger)J(\Phi)(U^\dagger\otimes U) = J(\Phi).
\end{equation}
These channels and their corresponding Choi matrices have recently been examined quite thoroughly in \cite{Johnston2019pcp,singh2020diagonal,singh2020ppt2,Singh2021triangle}. In particular, it has been shown that the action of a DUC (resp. CDUC) channel can be parameterized by two matrices $A,B\in \M{d}$ with equal diagonals:
\begin{alignat}{2}
    \Psi_{A,B}(X) &= \operatorname{diag}(A\ket{\operatorname{diag}X}) + \widetilde{B}\odot X^\top &&\equiv \sum_{i,j} A_{ij}X_{jj}\ketbra{i} + \sum_{i\neq j} B_{ij}X_{ji}\ketbra{i}{j} \label{eq:DUC}  \\ 
    \text{resp.} \quad \Phi_{A,B}(X) &= \operatorname{diag}(A\ket{\operatorname{diag}X}) + \widetilde{B}\odot X &&\equiv \sum_{i,j} A_{ij}X_{jj}\ketbra{i} + \sum_{i\neq j} B_{ij}X_{ij}\ketbra{i}{j} . \label{eq:CDUC}
\end{alignat}
where $\widetilde{B}=B-\operatorname{diag}B$ and $\odot$ denotes the entrywise or Hadamard matrix product.
\begin{remark}\label{remark:cptp-(C)DUC}
$\Phi_{A,B}$ (resp. $\Psi_{A,B}$) is completely positive and trace preserving if and only if
\begin{itemize}
    \item $A$ is entrywise non-negative (denoted $A\succcurlyeq 0$) and $B\geq 0$, see \cite[Lemma 6.11]{singh2020diagonal}.
    \item resp. $A\succcurlyeq 0$, $B=B^\dagger$, and $A_{ij}A_{ji}\geq |B_{ij}|^2$ for all $i,j\in [d]$, see \cite[Lemma 6.10]{singh2020diagonal}.
\end{itemize}
\end{remark}
A plethora of different classes of quantum channels can be shown to belong to the families of DUC and CDUC channels; see \cite[Section 7]{singh2020diagonal}. In particular, the previously considered classes of depolarizing- and MAD channels are CDUC, while the transpose-depolarizing channels are DUC. In this section, we derive very general sufficient conditions on the matrices $A,B\in\M{d}$ which ensure that the associated DUC and CDUC channels have positive complementary quantum capacity. Theorems~\ref{theorem:depol-transp} and \ref{theorem:MAD} would then pop out as special cases of these more general results. To this end, let us first see how the minimal output and environment dimensions of these channels can be obtained from the associated matrices $A,B\in\M{d}$. 
 
\begin{lemma}\label{lemma:(C)DUC-min}
For a DUC channel $\Psi_{A,B}:\M{d}\to\M{d}$ and a CDUC channel $\Phi_{A,B}:\M{d}\to \M{d}$,
\begin{align}
    d^*_{\out}(\Psi_{A,B}) &= |\{i\in [d] \, : \, \exists j\in [d], A_{ij}\neq 0 \}| = d^*_{\out}(\Phi_{A,B}) \\
    d^*_{\env}(\Psi_{A,B}) &= |\{i \in [d] \, : \, A_{ii} \neq 0\}| + \sum_{i<j} \operatorname{rank} \begin{bmatrix} A_{ij} & B_{ij} \\ B_{ji} & A_{ji} \end{bmatrix}  \\
    d^*_{\env}(\Phi_{A,B}) &= \operatorname{rank} B  + |\{ (i,j) \in [d]\times [d] \, : \, i \neq j \text{ and } A_{ij} \neq 0\}|
\end{align}
\end{lemma} 
\begin{proof}
Apply the channels $\Phi_{A,B}$ and $\Psi_{A,B}$ on the identity matrix $\iden_d\in\M{d}$ to conclude that their minimal output dimensions are both equal to the number of non-zero rows present in $A$, see Lemma~\ref{lemma:minimal}. The expressions of the minimal environment dimensions can be derived by computing the relevant Choi ranks, see \cite[Theorem 6.4 and Corollary 4.2]{singh2020diagonal} and Lemma~\ref{lemma:minimal}.
\end{proof}

Equipped with the necessary background, we are now ready to characterize the class of (C)DUC channels $\Phi$ for which there exist a pure state $\ketbra{\psi}$ such that $\operatorname{rank}\Phi(\ketbra{\psi})=\min \{d^*_{\out}(\Phi),d^*_{\env}(\Phi)\}$, so that Corollary~\ref{corollary:vikesh} can be invoked to obtain positivity of the (complementary) capacities of these channels whenever $d^*_{\out}(\Phi)\neq d^*_{\env}(\Phi)$. Recall the discussion following Remark~\ref{remark:vikesh}, where the difficulty in obtaining such a characterization for arbitrary channels is clearly highlighted. More often than not, one can only hope to obtain some sufficient conditions on the channel under consideration (like the one presented in Theorem~\ref{theorem:pauli}) to guarantee that the required pure state exists. This fact makes the following Proposition all the more important, since it contains a simple necessary and sufficient condition to ensure the existence of the desired pure state for arbitrary (C)DUC channels. 
 
\begin{proposition}\label{prop:CDUC-rank}
Let $\Phi:\M{d}\to\M{d}$ be a DUC/CDUC channel and $\ket{e}=\sum_{i}\ket{i}\in\C{d}$. Then,
\begin{equation}
    \max_{\ketbra{\psi}\in\St{d}}\operatorname{rank}\Phi(\ketbra{\psi})= \operatorname{rank}\Phi(\ketbra{e}{e}).
\end{equation}
\end{proposition}

\begin{proof} In order to prove the proposition, it suffices to show that for a DUC or a CDUC channel $\Phi$, if $\operatorname{rank}\Phi(\ketbra{e})<r$, then for all $\ketbra{\psi}\in\St{d}$, $\operatorname{rank}\Phi(\ketbra{\psi})<r$ as well. Moreover, since a positive semi-definite matrix $X\in\M{d}$ has $\operatorname{rank}X<r$ if and only if all its $r\times r$ principal submatrices are singular, we can choose to work with $r\times r$ principal submatrices corresponding to an arbitrary selection of $r$ rows and columns indexed by some $I\subseteq [d]$ with $|I|=r$. Hence, we only need to prove that (see the Methods section~\ref{appen:CDUC} for the relevant notation):
\begin{equation}
    \Phi(\ketbra{e})[I] \text{ is singular} \implies \Phi(\ketbra{\psi})[I] \text{ is singular for all } \ketbra{\psi}\in\St{d}.
\end{equation}
The following proof is split into two cases, and can be found in the Methods section~\ref{appen:CDUC}. 
\end{proof}

\begin{theorem}\label{theorem:CDUC-cap}
Let $\ket{e}=\sum_{i}\ket{i}\in\C{d}$ and $\Phi:\M{d}\to\M{d}$ be a DUC or a CDUC channel such that
\begin{equation}
\max_{\ketbra{\psi}\in\St{d}}\operatorname{rank}\Phi(\ketbra{\psi})= \operatorname{rank}\Phi(\ketbra{e}{e})=\operatorname{min}\{d^*_{\out}(\Phi),d^*_{\env}(\Phi)\}.    
\end{equation}
Then, 
\begin{itemize}
    \item if $d^*_{\out}(\Phi)>d^*_{\env}(\Phi)$, the channel $\Phi$ has positive quantum capacity.
    \item if $d^*_{\out}(\Phi)<d^*_{\env}(\Phi)$, any complementary channel $\Phi_c\in\mathcal{C}_\Phi$ has positive quantum capacity.
\end{itemize}
\end{theorem}
\begin{proof}
The proof is a straighforward consequence of Proposition~\ref{prop:CDUC-rank} and Corollary~\ref{corollary:vikesh}.
\end{proof}

\subsection{Generalized dephasing or Hadamard channels}

{Generalized dephasing} or {Hadamard} channels (also known as {Schur multipliers}) $\Phi_B:\M{d}\to \M{d}$ have the effect of diminishing the magnitude of the off-diagonal entries of the input states while perfectly preserving their diagonal elements. These channels are parameterized by a {correlation} matrix $B\in\M{d}$ (i.e., $B\geq 0$ and all its diagonal entries are equal to one):
\begin{equation}
    \forall X\in \M{d}: \qquad \Phi_B(X) = B\odot X.
\end{equation}
Note that $\odot$ above denotes the Hadamard or entrywise matrix product. It is clear that a dephasing channel $\Phi_B=\Phi_{A,B}$ is CDUC (see Eq.~\eqref{eq:CDUC}) with $A=\iden_d$. Complementary channels to the dephasing channels are known to be entanglement-breaking (and hence also antidegradable); so they have zero quantum capacity. This, in particular, implies that the dephasing channels themselves are all degradable and hence their quantum capacities admit a single-letter expression in terms of the channels' coherent information $\mathcal{Q}(\Phi_B)=\mathcal{Q}^{(1)}(\Phi_B)$. In the following theorem, we obtain a complete characterization of the set of zero capacity dephasing channels.

\begin{theorem}\label{theorem:dephasing}
Let $\Phi_B:\M{d}\to \M{d}$ be a dephasing channel associated with $B\in\M{d}$. Then,
\begin{equation*}
    \mathcal{Q}(\Phi_B)=0 \iff \Phi_B \text{ is entanglement-breaking }\iff \Phi_B \text{ is PPT }\iff B = \iden_d.
\end{equation*}
\end{theorem}
\begin{proof}
It is well-known that the final three equivalences hold, see \cite[Example 7.3]{singh2020diagonal} for instance. Moreover, if $\Phi_B$ is entanglement-breaking, then it clearly has zero capacity. Hence, to prove the theorem, it suffices to show that if $\Phi_B$ is not entanglement-breaking, then $\mathcal{Q}(\Phi_B)>0$. Now, since $\Phi_B$ is unital, $d^*_{\out}(\Phi_B)=d$ (see Remark~\ref{remark:minimal}). Moreover, it is easy to see that the Choi matrix $J(\Phi_B)=\sum_{i,j=0}^{d-1}B_{ij}\ketbra{ii}{jj}$ has the same rank as $B$, so that according to Lemma~\ref{lemma:minimal}, $d^*_{\env}(\Phi_B)=\operatorname{rank}B$. Let us split the remaining proof into two cases.\\

\noindent\textbf{Case 1.} $\operatorname{rank}B<d$. \\
Since $d^*_{\env}(\Phi_B)<d^*_{\out}(\Phi_B)$, the result follows by noting that $\Phi_B(\ketbra{e})=B$, see Theorem~\ref{theorem:CDUC-cap}. \\

\noindent\textbf{Case 2.} $\operatorname{rank}B=d$. \\
In this case, we have $d^*_{\env}(\Phi_B)=d^*_{\out}(\Phi_B)=d$, so Corollary~\ref{corollary:vikesh} can no longer be applied. Let $\{\ket{\phi_i}\}_{i=0}^{d-1}\subseteq \C{d}$ be a set of $d$ linearly independent unit vectors such that $B^\top=\operatorname{gram}\{\ket{\phi_i}\}_{i=0}^{d-1}$, i.e., $\forall i,j: B_{ij}=\langle \phi_j|\phi_i\rangle$. Then, we can define an isometry $V:\C{d}\to \C{d}\otimes \C{d}$ by setting $V\ket{i}=\ket{i}\otimes \ket{\phi_i}$ as its action on the standard basis of $\C{d}$, so that
\begin{equation}
    \forall X\in\M{d}: \qquad \Phi_B(X) = \Tr_2 (VXV^\dagger) \quad \text{and} \quad \Phi^c_B(X) = \Tr_1 (VXV^\dagger),
\end{equation}
where $\operatorname{Tr}_{1(2)}$ denote the partial trace over the first (second) system and $\Phi^c_B\in\mathcal{C}_{\Phi_B}$.
\medskip
Now, since $\Phi$ is not entanglement breaking and hence $B\neq \iden_d$, there exist indices $i\neq j$ such that $B_{ij}\neq 0$. Choose $\ket{\psi}=\ket{i}$, so that in the notation of Theorem~\ref{theorem:zero-necessary}, $R_\psi=\ketbra{i}{i}$ and $R^c_\psi=\ketbra{\phi_i}{\phi_i}$. Let us now analyse the action of the adjoint maps $\Phi^*_B$ and $(\Phi^c_B)^*$ on the projectors $R_\psi$ and $R^c_\psi$: 
\begin{alignat}{2}
    V\Phi^*_B(R_\psi) V^\dagger &= VV^\dagger (R_\psi\otimes \iden_d) VV^\dagger  \qquad\qquad V(\Phi^c_B)^*(R^c_\psi) V^\dagger &&= VV^\dagger (\iden_d \otimes R^c_\psi) VV^\dagger \nonumber \\
    &=\sum_{k=0}^{d-1} VV^\dagger \ketbra{ik}{ik}VV^\dagger &&= \sum_{k=0}^{d-1} VV^\dagger \ketbra{k\phi_i}{k\phi_i}VV^\dagger \nonumber \\
    &= \ketbra{i\phi_i}{i\phi_i} \sum_{k=0}^{d-1}|\langle\phi_i|k\rangle|^2 &&=\sum_{k=0}^{d-1} \ketbra{k\phi_k}{k\phi_k} |\langle \phi_k|\phi_i\rangle|^2. 
\end{alignat}
In the above calculation, we have used the fact that $VV^\dagger = \sum_{l=0}^{d-1} \ketbra{l\phi_l}{l\phi_l}$ acts as the orthogonal projection onto $\operatorname{range}V\subseteq \C{d}\otimes \C{d}$. As $\langle\phi_j|\phi_i\rangle\neq 0$, the number of terms in the second sum above is at least two, so that $V(\Phi^c_B)^*(R^c_\psi) V^\dagger$ has rank at least two. Therefore, $V(\Phi^c_B)^*(R^c_\psi) V^\dagger \not\leq V\Phi^*_B(R_\psi) V^\dagger \implies (\Phi^c_B)^*(R^c_\psi) \not\leq \Phi^*_B(R_\psi)$, since $V\Phi^*_B(R_\psi) V^\dagger$ is of unit rank. Theorem~\ref{theorem:zero-necessary} then tells us that $\mathcal{Q}^{(1)}(\Phi_B)>0$, and the proof is complete.
\end{proof}

\subsection{A family of channels with equal output and environment dimensions}
\label{subsec:equaloutenv}
Consider a unitary operator $V:{\mathbb{C}}^4 \to {\mathbb{C}}^2 \otimes {\mathbb{C}}^2$, and a pure state $\ket{\psi} \in {\mathbb{C}}^4 $ such that $V\ket{\psi} = \ket{0} \otimes \ket{1},$ where $\{\ket{0},\ket{1}\}$ denotes the computational basis of ${\mathbb{C}}^2 $. Then, for the quantum channel $\Phi:\M{4}\rightarrow \M{2}$ and its complement $\Phi_c:\M{4}\rightarrow \M{2}$ defined as follows, $d^{*}_{\out}(\Phi) = d^{*}_{\env}(\Phi)=2$:  
\begin{equation}
    \forall X\in\M{4}: \qquad \Phi(X) = \Tr_2 (VXV^\dagger) \quad, \quad \Phi_c(X) = \Tr_1 (VXV^\dagger).
\end{equation}
the orthogonal projections onto $\operatorname{range}\Phi(\ketbra{\psi})$ and $ \operatorname{range}\Phi_c(\ketbra{\psi})$ are given by $R_\psi = \ketbra{0}$ and $R_\psi^c = \ketbra{1},$
respectively. It is clear that the adjoints of $\Phi$ and $\Phi_c$ act on these projectors as follows
\begin{equation}
    \Phi^*(R_\psi) = V^\dagger (R_\psi \otimes {\iden}_2) V \quad , \quad \Phi_c^*(R_\psi^c) = V^\dagger ({\iden}_2 \otimes R_\psi^c) V.
\end{equation}
Now, by using the fact that $VV^\dagger = {\iden}_4$ (since $V$ is unitary), we obtain
\begin{align}
    V\Phi^*(R_\psi)V^\dagger &= VV^\dagger (R_\psi \otimes {\iden}_2) VV^\dagger = (R_\psi \otimes {\iden}_2)= \ketbra{0} \otimes (\ketbra{0} + \ketbra{1}),\nonumber\\
    {\hbox{and}} \quad  V\Phi_c^*(R_\psi^c)V^\dagger
&= VV^\dagger (\iden_2 \otimes R^c_\psi ) VV^\dagger = ({\iden}_2 \otimes R^c_\psi)= (\ketbra{0} + \ketbra{1}) \otimes \ketbra{1}.
\end{align}
Hence, $V\Phi_c^*(R^c_\psi)V^\dagger \not\leq V\Phi^*(R_\psi)V^\dagger$, and therefore $\Phi_c^*(R^c_\psi)\not\leq \Phi^*(R_\psi)$. By Theorem~\ref{theorem:zero-necessary},
we then conclude that such a channel $\Phi$ has $\mathcal{Q}^{(1)}(\Phi) >0$. The same argument applied to the complementary channel shows that $\mathcal{Q}^{(1)}(\Phi_c) > 0$ as well.

By considering unitary mappings $V:\mathbb{C}^{d^2}\to \mathbb{C}^d\otimes \mathbb{C}^d$ and following the same steps as above, it is easy to infer that for any channel $\Phi:\mathcal{M}_{d^2}\to \M{d_{\out}}$ with $d = d^{*}_{\out}(\Phi) = d^{*}_{\env}(\Phi)$, both $\mathcal{Q}^{(1)}(\Phi)$ and $\mathcal{Q}^{(1)}(\Phi_c)$ are positive. In particular, such a channel can neither be degradable nor anti-degradable.

\subsection{Structure theorems for more capable quantum channels} \label{sec:morecapable}
Cubitt et al proved two structure theorems for degradable channels (see ~\cite[Theorem 3 and 4]{Cubitt2008degradable}), which we extend in this section to the strictly larger class of more capable quantum channels. Their proofs follow simply from our results in Section~\ref{sec:main}. 

\begin{theorem}
\label{thm:V.1}
Let $\Phi:\M{d}\rightarrow \M{d_{\out}}$ be a more capable quantum channel. If there exists a pure state $\ketbra{\psi} \in \cS_d$ such that $\operatorname{rank} \Phi (\ketbra{\psi}) = d^{*}_{\out} (\Phi)$, then $d^{*}_{\out} (\Phi) = d^{*}_{\operatorname{env}} (\Phi)$.
\end{theorem}

\begin{proof}
By Lemma~\ref{lemma:2}, we known that for all $\ketbra{\psi}\in\St{d}$,  ${\rm{rank}} \,\Phi(\ketbra{\psi}{\psi}) = {\rm{rank}}\, \Phi_c(\ketbra{\psi}{\psi})$. Hence, we must have
$d^{*}_{\operatorname{env}} (\Phi)  \geq d^{*}_{\out} (\Phi)$. However, if 
$d^{*}_{\operatorname{env}} (\Phi) > d^{*}_{\out} (\Phi)$, then according to Corollary~\ref{corollary:vikesh}, $\mathcal{Q}(\Phi_c) >0$, which leads to a contradiction. Hence, $d^{*}_{\operatorname{env}} (\Phi) = d^{*}_{\out} (\Phi)$.
\end{proof}

\begin{theorem}
\label{thm:V.2}
Let $\Phi:\M{d}\rightarrow \M{d_{\out}}$ be a more capable quantum channel with $d^{*}_{\out} (\Phi)=2$. Then $d^{*}_{\operatorname{env}} (\Phi) \leq 2$ and $d \leq 3$.
\end{theorem}

\begin{proof}
If there exists a pure state $\ketbra{\psi} \in \M{d}$ such that ${\rm{rank}} \,\Phi(\ketbra{\psi}{\psi}) =2$, then Theorem~\ref{thm:V.1} implies that $d^{*}_{\operatorname{env}} (\Phi)=2$. If no such pure state exists, then it must be the case that $\Phi$ always maps pure states to pure states. Then by \cite[Theorem 1]{Cubitt2008degradable}, either $(i)$ $\Phi(X)= UX U^\dagger$, where $U:\C{d}\to \C{d_{\out}}$ is an isometry, 
in which case Lemma~\ref{lemma:minimal} tells us that $d^{*}_{\operatorname{env}} (\Phi)=1$, or $(ii)$ $\Phi$ is the completely noisy channel which maps everything to a single fixed pure state: $\Phi(X)
= \Tr(X) \ketbra{\phi}$, which cannot be true since $d^{*}_{\out} (\Phi)=2$. Hence, we always have $d^*_{\env}(\Phi)\leq 2$.
Now, let $V:\C{d}\to \C{2}\otimes \C{d_{\env}}$ be the Stinespring isometry which defines the minimally defined complementary pair $\Phi:\M{d}\to\M{2}$ and $\Phi_c:\M{d}\to \M{d_{\env}}$ (i.e. $d_{\env}=d^*_{\env}(\Phi)$), see Eq.~\eqref{eq:complementary}. Then, 
\begin{equation}
    d \leq 2 d^{*}_{\env} (\Phi) \implies d \leq \begin{cases}
    2 &\text{if } d^{*}_{\env} (\Phi)=1 \\
    4 &\text{if } d^{*}_{\env} (\Phi)=2
    \end{cases}.
\end{equation}
If $d=4$, $V: \C{4} \rightarrow \C{2}\otimes \C{2}$ becomes a unitary map, which implies that $\mathcal{Q}^{(1)}(\Phi_c)>0$ and hence leads to a contradiction (see Section~\ref{subsec:equaloutenv}). Thus, $d \leq 3$ and the proof is complete. 
\end{proof}

\section{Discussion}
In this paper, we have employed some basic techniques from analytic perturbation theory of Hermitian matrices (Theorem~\ref{theorem:perturbation}) to derive a powerful sufficient condition for a quantum channel (or its complement) to have positive quantum capacity (Theorem~\ref{theorem:main}). An equivalent formulation of this condition equips us with a first-of-its-kind necessary condition for membership in the set of zero capacity quantum channels (Theorem~\ref{theorem:zero-necessary}). These are significant results because to date, no systematic procedure is known to check if a given quantum channel can be used to reliably transmit quantum information, and this is precisely because of our limited understanding of the set of zero capacity quantum channels. Notably, the main result in \cite{siddhu2020logsingularities} pops out as an immediate consequence of Theorem~\ref{theorem:main}, which can be used to detect positive quantum capacities of channels (or their complements) with unequal output and environment dimensions for which there exists a pure input state which gets mapped to a maximal rank output state (Corollary~\ref{corollary:vikesh}). Obtaining a complete characterization of channels for which such a pure state exists is a rather formidable task. Nevertheless, we have derived simple inequalities between the input, output and environment dimensions of a given channel that suffice to ensure the existence of the sought-after pure state, irrespective of the specific structure of the channel (Corollaries~\ref{corollary:vikesh2} and \ref{corollary-biginput}). By exploiting our main results, we have shown that a variety of interesting examples of quantum channels have positive quantum capacity. Moreover, our results lead to simplified proofs of certain existing structure theorems for the class of degradable quantum channels, and an extension of their applicability to the larger class of more capable quantum channels.
 
Listed below are some of the open problems that stem from our research.

\begin{itemize}
    \item Show that PPT channels satisfy the necessary condition of Theorem~\ref{theorem:zero-necessary} for membership in the set of zero capacity quantum channels.
    \item Look for other quantum channels which satisfy the condition stated in Theorem~\ref{theorem:zero-necessary}. This could potentially lead to new examples of channels which are neither anti-degradable nor PPT, but still have zero quantum capacity.
    \item Obtain examples of channels which satisfy the condition stated in Theorem~\ref{theorem:zero-necessary} for all positive integers $n\leq m$ for some fixed $m\in\mathbb{N}$, but not for $n>m$. Such channels, if they exist, would have positive $n$-shot coherent information for all $n>m$ but at the same time, would satisfy the necessary condition for having zero $n$-shot coherent information for all $n\leq m$, and hence could potentially shed light on the superadditivity of coherent information.
    \item Investigate whether Theorem~\ref{theorem:main} can be used to detect positive quantum capacity of a (C)DUC channel when Theorem~\ref{theorem:CDUC-cap} is inapplicable.
    \item Check if the perturbative techniques applied in this paper can be extended to not only detect positivity of quantum channel capacities but also to provide meaningful lower bounds on the quantum capacities.
\end{itemize}

In \cite{Singh2022coherent}, we apply the techniques developed in this paper to random quantum channels to show that typically, the coherent information of a single copy of a randomly selected channel is guaranteed to be positive if the channel's output space is larger than that of its environment. This is very interesting since in general, it is known that the coherent information needs to be computed for an unbounded number of copies of a given channel in order to detect its capacity \cite{Cubitt2015unbounded}. Hence, whenever the channel's output space is larger than the environment, we can be almost sure that a single copy of the channel has positive coherent information. 

\section{Methods}

\subsection{Minimal output and environment dimensions}\label{appen:lemmaminimal}

\noindent \textit{Proof of Lemma~\ref{lemma:minimal}.}\hspace{0.2cm}Let us begin with the expression for the minimal environment dimension. For any channel $\Phi:\M{d}\to \M{d_{\out}}$, the Stinespring dilation theorem tells us that there exists an isometry $V:\C{d}\to \C{d_{\out}}\otimes \C{d_{\env}}$ with $\operatorname{rank}J(\Phi)= d_{\env}$ such that $\Phi(X)=\operatorname{Tr}_{\env}(VXV^\dagger)$ for all $X\in\M{d}$, see \cite[Corollary 2.27]{watrous2018theory}. Moreover, since any other isometry $V:\C{d}\to \C{d_{\out}}\otimes \C{d_{\env}}$ that defines $\Phi$ as above gives a rank one decomposition of its Choi matrix:
\begin{equation}
    J(\Phi) = \sum_{k=1}^{d_{\env}} \ketbra{ \operatorname{vec}(V_k)}, \quad\text{where}\,\, V_k = (\iden_d \otimes \bra{k})V,
\end{equation}
it is clear that $\operatorname{rank}J(\Phi)\leq d_{\env}$ for all such isometries $V:\C{d}\to \C{d_{\out}}\otimes \C{d_{\env}}$. This shows that $d^{*}_{\operatorname{env}}(\Phi)=\operatorname{rank}J(\Phi)$. 

Now, let $r=\operatorname{rank}\Phi'_c(\iden_d)$ for some $\Phi'_c\in\mathcal{C}_\Phi$. Then, since all complementary channels $\Phi_c:\M{d}\to\M{d_{\env}}$ are isometrically related (see Remark~\ref{remark:iso-exten}), we have $r=\operatorname{rank}\Phi_c(\iden_d)\leq d_{\env}$ for all such $\Phi_c:\M{d}\to\M{d_{\env}}$, so that $r\leq d^*_{\env}(\Phi)$. To establish the reverse inequality, observe that for any positive semi-definite $X\in\M{d}$ and $\Phi_c:\M{d}\to\M{d_{\env}}$ complementary to $\Phi$, there exists a small enough $\delta>0$ such that 
\begin{equation}
    \delta X\leq \iden_d\implies \delta\Phi_c(X)\leq \Phi_c(\iden_d) \implies \operatorname{range}\Phi_c(X)\subseteq \operatorname{range} \Phi_c(\iden_d) \simeq \C{r}.
\end{equation}
If $W:\operatorname{range}\Phi_c(\iden_d)\to \C{d_{\env}}$ denotes an isometric embedding of $\operatorname{range}\Phi_c(\iden_d)$ into $\C{d_{\env}}$, then the channel $\Phi'_c:\M{d}\to\M{r}$ defined as $\Phi'_c(X) = W^\dagger\Phi_c(X) W$ is also complementary to $\Phi$, so that $d^*_{\env}(\Phi)\leq r$. This proves that $d^*_{\env}(\Phi)=r$. 

Finally, the expression for the minimal output dimension can be derived by applying the formula for the minimal environment dimension to any complementary channel $\Phi_c\in\mathcal{C}_\Phi$. \qed

\subsection{Analytic perturbation theory and the proof of Theorem~\ref{theorem:main}}\label{appendix:main}

The main tool that we employ in the proof of Theorem~\ref{theorem:main} is the following seminal result from analytic perturbation theory.

\begin{theorem}\label{theorem:perturbation}
\cite[Introduction, Theorem 1]{baum1985perturbation} or \cite[Chapter 1, Theorem 1]{rellich1969perturbation}\\Let $A(\epsilon)=A_0+\epsilon A_1$, where $A_0,A_1\in \M{d}$ are Hermitian matrices and $\epsilon\in \mathbb{R}$. Let $\lambda_0$ be an eigenvalue of $A_0$ of multiplicity $n$ and let $P_0$ be the orthogonal projection onto the corresponding $n-$dimensional eigenspace. Then, for sufficiently small $|\epsilon|$, $A(\epsilon)$ has exactly $n$ analytic eigenvalues (counted with multiplicities) in the neighborhood of $\lambda_0$ with convergent power series expansions:
\begin{equation}
    \lambda_i(\epsilon) = \lambda_0 + \lambda_{i1}\epsilon + \lambda_{i2}\epsilon^2 + \ldots, \quad i=1,2,\ldots, n .
\end{equation}
Further, the non-zero first order correction terms $\lambda_{i1}$ are the non-zero eigenvalues of $P_0 A_1 P_0$.
\end{theorem}

Before proving our main theorem, we need to develop a few auxiliary lemmas.

\begin{lemma} \label{lemma:1}
Let $R>0$ and $f\in C^\omega [0,R)$ be a non-negative real analytic function (not identically equal to zero) with the following power series expansion:
\begin{equation}
    f(\epsilon) = f_0 + f_1 \epsilon + f_2 \epsilon ^2 + \ldots, \quad \epsilon\in [0,R)
\end{equation}
Then, there exists $r\in (0,R)$ such that $g(\epsilon)=f(\epsilon)\log f(\epsilon)$ is real analytic on the interval $(0,r)$, i.e. $g\in C^\omega(0,r)$. Moreover, the derivative $g'\in C^\omega (0,r)$ is unbounded if and only if $f_0 = 0$ and $f_1\neq 0$.
\end{lemma}

\begin{proof}
This proof exploits some basic properties of real analytic functions, see \cite[Chapter 1]{krantz2002primer}. Firstly, recall that since the Taylor series of $f$ converges in the interval $[0,R)$, the derived series $f'$ also converges in the same interval. Hence, for every $r\in (0,R)$ we can uniformly bound both $f$ and $f'$ on the domain $[0,r]$. Moreover, since the zeros of analytic functions are isolated, it is easy to find an $r \in (0,R)$ such that no zeros of $f$ lie within the range $(0,r]$, i.e. $\forall \epsilon\in (0,r]: f(\epsilon)>0$. Hence, there exist constants $c_1,c_2>0$ such that $\forall \epsilon\in (0,r): \,\, c_1\leq f(\epsilon)\leq c_2$, except possibly when $f_0=f(0)=0$, in which case $f$ is still strictly positive and bounded on $(0,r)$ but $f(\epsilon)\rightarrow 0$ as $\epsilon\rightarrow 0$.

Now, since $\log$ is real analytic on $(0,\infty)$ and $f$ maps $(0,r)$ within $(0,\infty)$, it is clear that $g$ is real analytic on $(0,r)$.  Then, $g'$ is also analytic on $(0,r)$ and admits an expression of the form
\begin{align}
    g'(\epsilon) &= f'(\epsilon)\log f(\epsilon) + f'(\epsilon) \nonumber\\
          &= (f_1 + \mathcal{O}(\epsilon)) \log (f_0 + f_1\epsilon + \mathcal{O}(\epsilon^2)) + f_1 + \mathcal{O}(\epsilon) \nonumber\\ 
          &= f_1 \log (f_0 + f_1\epsilon + \mathcal{O}(\epsilon^2)) + \left[\mathcal{O}(\epsilon) \log (f_0 + f_1\epsilon + \mathcal{O}(\epsilon^2)) + f_1 + \mathcal{O}(\epsilon)\right],
\end{align}
where we have used the symbol $\mathcal{O}(\epsilon^k)$ as a placeholder for functions $h$ for which there exists $0< r_0 \leq r$ and $C>0$ such that $\forall \epsilon\in [0,r_0): \, |h(\epsilon)|\leq C\epsilon^k$. It is then straightforward to infer from our previous discussion that the terms in the square brackets above are always bounded in the interval $(0,r)$. Now, if either $f_0\neq 0$ or $f_1=0$, the first term is also bounded. However, if $f_0=0$ and $f_1\neq 0$, the first term splits up into $f_1\log \epsilon + f_1\log (f_1+\mathcal{O}(\epsilon))$, which blows up as $\epsilon\rightarrow 0$. Hence, we can write the derivative as follows
\begin{equation}\label{eq:derivative}
    g'(\epsilon) = \begin{cases} 
            f_1\log \epsilon + K_1(\epsilon), \quad &\text{if } f_0=0,f_1\neq 0 \\
            K_2(\epsilon), \quad &\text{otherwise}
            \end{cases}
\end{equation}
where both $K_1$ and $K_2$ are bounded on $(0,r)$.
\end{proof}

\begin{lemma}\label{lemma:2}
Let $\Phi:\M{d}\rightarrow \M{d_{\out}}$ be a channel and $\Phi_c\in \mathcal{C}_\Phi$. Then, for a pure state $\ketbra{\psi}{\psi}\in \mathcal{S}_d$, the non-zero eigenvalues (counted with multiplicity) of $\Phi(\ketbra{\psi}{\psi})$ and $\Phi_c(\ketbra{\psi}{\psi})$ are identical.
\end{lemma}
\begin{proof}
Let $V:\C{d}\rightarrow \C{d_{\out}}\otimes \C{d_{\operatorname{env}}}$ be the isometry that defines the complementary channels $\Phi$ and $\Phi_c$ as in Eq~\eqref{eq:complementary}. Using Schmidt decomposition, we can express $V\ket{\psi}=\sum_{i=1}^n \sqrt{s_i}\ket{i}_{\out}\ket{i}_{\operatorname{env}}$, where $\{\ket{i}_{\out}\}_{i=1}^n\subseteq \C{d_{\out}}$ and $\{\ket{i}_{\operatorname{env}}\}_{i=1}^n\subseteq \C{d_{\operatorname{env}}}$ are orthonormal sets, $s_i\geq 0$ and $n\leq \min \{d_{\out},d_{\env}\}$. Then, the desired result follows by noting that 
\begin{equation}
    \Phi(\ketbra{\psi}{\psi})=\sum_i s_i\ket{i}_{\out}\bra{i} \qquad \text{and}\qquad \Phi_c(\ketbra{\psi}{\psi}) = \sum_{i} s_i \ket{i}_{\operatorname{env}}\bra{i}.
\end{equation}
\end{proof}

We are now ready to prove Theorem~\ref{theorem:main}.

\noindent\textit{Proof of Theorem~\ref{theorem:main}}.
We consider the one-parameter family of states $\rho (\epsilon) = (1-\epsilon)\ketbra{\psi}{\psi} + \epsilon \sigma$ within the range $0\leq \epsilon \leq 1$, and analyse the outputs under the given channels 
\begin{alignat}{3}
\Phi[\rho(\epsilon)] &= (1-\epsilon)\Phi(\ketbra{\psi}{\psi}) + \epsilon \Phi(\sigma), \qquad\qquad\qquad \Phi_c[\rho(\epsilon)]&&=(1-\epsilon)\Phi_c(\ketbra{\psi}{\psi}) + \epsilon \Phi_c(\sigma). \nonumber\\
&= \Phi(\ketbra{\psi}{\psi}) + \epsilon [\Phi(\sigma)-\Phi(\ketbra{\psi}{\psi})]  &&= \Phi_c(\ketbra{\psi}{\psi}) + \epsilon [\Phi_c(\sigma)-\Phi_c(\ketbra{\psi}{\psi})]
\end{alignat}
At $\epsilon=0$, both $\Phi(\ketbra{\psi}{\psi})$ and $\Phi_c(\ketbra{\psi}{\psi})$ have the same non-zero spectrum (Lemma~\ref{lemma:2}):
\begin{equation}
    \operatorname{spec}\Phi(\ketbra{\psi}{\psi}) = \{s_{10},\ldots ,s_{n0},0,\ldots ,0 \} \quad\text{and}\quad \operatorname{spec}\Phi_c(\ketbra{\psi}{\psi}) = \{s_{10},\ldots ,s_{n0},0,\ldots ,0 \},
\end{equation}
where $n\in \mathbb{N}$ is the number of non-zero eigenvalues counted with multiplicities, so that $K_\psi$ and $K^c_\psi$ are projectors onto the respective null spaces of dimensions $\kappa = d_{\out}-n$ and $\kappa_c = d_{\operatorname{env}}-n$.
Observe that for our purposes, the spectrum of a $d\times d$ matrix is just the (unordered) sequence of all its eigenvalues counted with multiplicities. Now, once we turn on the perturbation and $\epsilon$ is sufficiently small (say $0\leq\epsilon<R $), we get the following analytic eigenvalue functions (see Theorem~\ref{theorem:perturbation}):
\begin{alignat}{2}
    s_i(\epsilon) &= s_{i0}+s_{i1}\epsilon + s_{i2}\epsilon^2 \ldots \qquad\qquad \lambda_j (\epsilon) &&= 0+ \lambda_{j1}\epsilon + \lambda_{j2}\epsilon^2\ldots \nonumber\\ 
    s^c_i(\epsilon) &= s_{i0}+s^c_{i1}\epsilon +s^c_{i2}\epsilon^2 \ldots \qquad\qquad \lambda^c_k (\epsilon) &&= 0+ \lambda^c_{k1}\epsilon +\lambda^c_{k2}\epsilon^2\ldots
\end{alignat}
such that $$\operatorname{spec}\Phi(\rho(\epsilon))=\{s_i(\epsilon), \lambda_j(\epsilon)\} \qquad\text{and}\qquad \operatorname{spec}\Phi_c(\rho(\epsilon))=\{s^c_i(\epsilon), \lambda^c_k(\epsilon)\},$$ where $i=1,2,\ldots ,n$, $j=1,2,\ldots ,d_{\out}-n$ and $k=1,2,\ldots ,d_{\operatorname{env}}-n$. Assume, without loss of generality, that all the above functions are not identically zero (otherwise, we can just restrict ourselves to those which are non-zero). We are interested in analyzing the coherent information of $\rho(\epsilon)$ with respect to $\Phi$, which has the following form (see Section~\ref{subsec:def})
\begin{align}
    I(\epsilon) := I_c(\rho(\epsilon);\Phi) &= S[\Phi(\rho(\epsilon))] - S[\Phi_c(\rho(\epsilon))] \nonumber\\ &=\sum_{i}s^c_i(\epsilon)\log s^c_i(\epsilon) + \sum_k \lambda^c_k(\epsilon) \log \lambda^c_k(\epsilon) - \sum_{i}s_i(\epsilon)\log s_i(\epsilon) - \sum_j \lambda_j(\epsilon) \log \lambda_j(\epsilon).
\end{align}
Observe that we can immediately apply Lemma~\ref{lemma:1} to each term in the above sum, so that we obtain $0<r<R$ such that $I(\epsilon)$ is analytic on $(0,r)$ and its derivative can be written as 
\begin{equation}
I'(\epsilon) = \left[ \sum_{j:\lambda_{j1}\neq 0} \lambda_{j1} - \sum_{k : \lambda^c_{k1}\neq 0} \lambda^c_{k1} \right] \log (1/\epsilon) + K(\epsilon),
\end{equation}
where $K$ is bounded on $(0,r)$, see Eq~\eqref{eq:derivative}. Now, Theorem~\ref{theorem:perturbation} tells us that the term in the brackets above is nothing but $\operatorname{Tr}(K_\psi \Phi(\sigma)) - \operatorname{Tr}(K^c_\psi \Phi_c(\sigma))$. Hence, if $\operatorname{Tr}(K_\psi \Phi(\sigma)) > \operatorname{Tr}(K^c_\psi \Phi_c(\sigma))$, we can find a small enough $0<\delta<r$ such that the first term above dominates the other when $\epsilon < \delta$ and consequently, $\forall \epsilon\in (0,\delta): I'(\epsilon)>0$. A simple application of the mean value theorem then tells us that $I(\epsilon)$ is strictly increasing on the interval $[0,\delta]$. Moreover, since $I(0)=0$, we obtain $\forall \epsilon\in (0,\delta]: I(\epsilon)>0 \implies \mathcal{Q}^{(1)}(\Phi) > 0$. The other case can be tackled similarly. \hspace{2.7cm} $\square$

\subsection{Matrix subspaces and the proof of Lemma~\ref{lemma:dim-rank-bound}}\label{appen:lemmadim-rank-bound}

\noindent\textit{Proof of Lemma~\ref{lemma:dim-rank-bound}} \hspace{0.2cm} This result was first proven in \cite{Flanders1962rank} for matrix spaces over fields $F$ with cardinality $|F|\geq r+1$, and was later generalized to work for arbitrary fields in \cite{Meshulam1985rank}. We nevertheless provide a simple proof for the case when $F=\mathbb{C}$ is the field of complex numbers. By padding all the matrices in $\mathcal{S}$ with extra zero rows, we can assume that $d_1=d_2=d$. Now, without loss of generality, we can further assume that the following block matrix is contained in $\mathcal{S}$
\begin{equation}
    I = \begin{bmatrix} \iden_r & 0 \\ 0 & 0 \end{bmatrix} \in \mathcal{S},
\end{equation}
where $\iden_r\in \M{r}$ is the identity matrix. If $\mathcal{S}$ doesn't contain such a matrix, then it is easy to find non-singular matrices $P,Q\in\M{d}$ such that the subspace $P\mathcal{S}Q:=\{PXQ : X\in\mathcal{S}\}$, which has the same dimension as $\mathcal{S}$, contains the aforementioned matrix. Now, let
\begin{equation}
    \widetilde{\mathcal{S}}:= \left\{\begin{bmatrix} 0 & B^\dagger \\ B & A \end{bmatrix} \in \M{d} : B\in\M{(d-r)\times r}, A\in \M{d-r} \right\}
\end{equation}
be a matrix subspace with $\dim \widetilde{\mathcal{S}} = r(d-r)+(d-r)^2=d(d-r)$. Then, for $0\neq X\in\mathcal{S}\cap\widetilde{\mathcal{S}}$,
\begin{equation}
    X+cI = \begin{bmatrix} c\iden_r & B^\dagger \\ B & A \end{bmatrix}\in \mathcal{S} \implies \operatorname{rank}(X+cI)\leq r,
\end{equation}
where the implication holds for all $c\in\mathbb{C}$. However, a simple application of the Schur complements of block matrices reveals that $\operatorname{rank}(X+cI)=\operatorname{rank}c\iden_r + \operatorname{rank}(A-\frac{B B^\dagger}{c})>r$, where $0\neq c\in\mathbb{C}$ is such that $A\neq \frac{B^\dagger B}{c}$. This leads to a contradiction, and hence $\mathcal{S}\cap \widetilde{\mathcal{S}}=\{0\}$. Hence, we arrive at the desired conclusion: $\dim (\mathcal{S}+\widetilde{\mathcal{S}}) = \dim\mathcal{S} + \dim \widetilde{\mathcal{S}} \leq d^2 \implies \dim\mathcal{S} \leq d^2 - d(d-r) = dr$. \qed

\subsection{An auxiliary lemma}
\begin{lemma}\label{lemma:3}
Let $\rho\in \St{d}$ and $R_\rho$ denote the orthogonal projection onto $\operatorname{range}\rho$. Then, for an arbitrary channel $\Phi:\M{d}\rightarrow \M{d_{\out}}$, the following equality holds: $\operatorname{range}\Phi(R_\rho)=\operatorname{range}\Phi(\rho)$. 
\end{lemma}
\begin{proof}
Spectrally decompose $\rho=\sum_{i=1}^n \lambda_i P_i$, where $\lambda_i>0$ are the distinct non-zero eigenvalues of $\rho$ and $P_i$ are the orthogonal projections onto the corresponding eigenspaces, so that $\Phi(R_\rho)=\sum_i \Phi(P_i)$ and $\Phi(\rho)=\sum_i \lambda_i \Phi(P_i)$. Now, the conclusion follows by noting that
\begin{align}
     \ker\sum_i \Phi(P_i) = \bigcap_i \ker\Phi(P_i) = \bigcap_i \ker\lambda_i \Phi(P_i) &= \ker \sum_i \lambda_i \Phi(P_i).  
\end{align}
\end{proof}

\subsection{(C)DUC channels and the Proof of Proposition~\ref{prop:CDUC-rank}} \label{appen:CDUC}
Before proving the proposition, it is essential to introduce some notation. Given a $d_1\times d_2$ matrix $X$ and $I\subseteq [d_1], J\subseteq [d_2]$, the $|I|\times |J|$ {submatrix} containing the rows and columns of $X$ corresponding to the indices $i\in I$ and $j\in J$, respectively, is denoted by $X[I|J]$. The basis vectors $\{\ket{i}\}_{i\in I}$ and $\{\ket{j}\}_{j\in J}$ are then thought to lie in $\C{|I|}$ and $\C{|J|}$, respectively, so that $X[I|J]=\sum_{i\in I, j\in J}X_{ij}\ketbra{i}{j}$. If $I=[d_1]$ (or $J=[d_2]$), the associated submatrix is denoted by $X[d_1|J]$ (or $X[I|d_2]$). If $I=J$, we use $X[I]:=X[I|I]$ to denote the associated $|I|\times |I|$ {principal} submatrix. Let us also recall the definition of the determinant function $\det :\M{d}\to \mathbb{C}$:
\begin{equation}\label{eq:det}
    \det (X) = \sum_{\pi\in S_d} \rm{sgn}(\pi) X_{\pi},
\end{equation}
where $S_d$ denotes the permutation group of the set $[d]$ and $X_\pi = \prod_{i\in [d]}X_{i\pi_i}$.

Finally, let us prove a lemma regarding an intriguing matrix determinant.

\begin{lemma}\label{lemma:det}
Let $X\in\M{d}$ $(d\geq 2)$ be such that it has exactly two non-zero entries in each row, i.e., $\forall i\in[d]$, there exist unique $i_1,i_2\in[d]$ such that $X_{ij}\neq 0 \iff j\in\{i_1,i_2\}$. For $\ket{\psi}\in\C{d}$, define
\begin{align}
    \psi(X)_{ij} = \begin{cases}
    0 &\text{ if } j\notin \{i_1,i_2\} \\
    X_{ii_1}\psi_{i_2}  &\text{ if } j=i_1 \\
    X_{ii_2}\psi_{i_1} &\text{ if } j=i_2
    \end{cases}
\end{align}
Then, 
\begin{equation}
\det \psi(X) = (\prod_{i\in [d]} \psi_i^{n_i - 1}) \det X,
\end{equation}
where $n_i$ denotes the number of non-zero entries in the $i^{th}$ column of $X$. 
\end{lemma}
\begin{proof}
Let us begin by observing that for every $\pi\in S_d$, we must have $\psi(X)_\pi = (\prod_{i\in [d]}\psi_i^{k_i(\pi)}) X_\pi$, where $\{k_i(\pi)\}_{i\in [d]}$ is a set of positive integers. This is because the matrix $\psi(X)$ has the same entries as $X$ upto multiplication by the entries of $\ket{\psi}$. Hence, our task is to show that 
\begin{equation}
\forall i\in [d], \forall \pi\in S_d: \quad k_i(\pi) = n_i-1,    
\end{equation}
where $n_i$ denotes the number of non-zero entries in the $i^{th}$ column. Firstly, note that if $\pi$ is such that $X_\pi=0$, then $\psi(X)_\pi=0$ as well, so there's nothing to prove here. Hence, let $\pi\in S_d$ be an arbitrary permutation with $X_\pi\neq 0$. Now, suppose on the contrary that the desired claim does not hold and there exists $j\in [d]$ such that $k_j(\pi)\neq n_j-1$. Let us collect all the rows $i\in [d]$ with $X_{ij}\neq 0$ in the set $R_j$, so that $n_j = |R_j|$. Then, for each row $i\in R_j$, one of the two non-zero entries is present in the column $i_1 = j$. Since $X_\pi\neq 0$, it must be the case that $\pi_i = i_2$ for all $ i\in R_j$ such that $\pi_i\neq j=i_1$, which implies that 
\begin{equation}
    \prod_{ \substack{i\in R_j \\  \pi_i\neq j}} \psi(X)_{i\pi_i} = \prod_{ \substack{i\in R_j \\  \pi_i\neq j}} X_{ii_2}\psi_{i_1} = \psi_j^{n_j -1} \prod_{ \substack{i\in R_j \\  \pi_i\neq j}} X_{i\pi_i}.
\end{equation}
Since $\psi_j$ cannot arise from any other term in the product $\prod_{i\in [d]}\psi(X)_{i\pi_i}$, we arrive at a contradiction.
\end{proof}

\noindent\textit{Proof of Proposition~\ref{prop:CDUC-rank}} \hspace{0.2cm} Recall that in order to prove the proposition, it suffices to show that 
\begin{equation}
    \Phi(\ketbra{e})[I] \text{ is singular} \implies \Phi(\ketbra{\psi})[I] \text{ is singular for all } \ketbra{\psi}\in\St{d},
\end{equation}
where $I\subseteq [d]$ is an arbitrary index set of size $|I|=r$. \\

\noindent\textbf{Case 1.} $\Phi=\Phi_{A,B}$ is a CDUC channel with $A\succcurlyeq 0$ and $B\geq 0$. \\
Let us express $\Phi_{A,B}(\ketbra{e}{e})[I]=\sum_{i\in I}\sum_{k\neq i} A_{ik}\ketbra{i} + B[I]$ as a sum of two $r\times r$ positive semi-definite matrices, so that $\Phi_{A,B}(\ketbra{e})[I]$ is singular if and only if the two matrices in the above sum have a non-trivial common kernel. Since $\ker\sum_{i\in I}\sum_{k\neq i} A_{ik}\ketbra{i} = \operatorname{span}\{\ket{j}\}_{j\in I\setminus J}\subseteq \C{r}$ with $J=\{j\in I \, : \, \exists k\in [d], k\neq j, A_{jk}\neq 0\ \}$, we obtain
\begin{equation}
    \Phi_{A,B}(\ketbra{e})[I] \, \text{ is singular} \iff \{\ket{B[I]_j} \}_{j\in I\setminus J}\subseteq \C{r}\, \text{ is linearly dependent},
\end{equation}
where $\ket{B[I]_j}:=B[I]\ket{j}$ is the $j^{th}$ column of $B[I]$. Now, for $\ket{\psi}\in\C{d}$, we have
\begin{equation}\label{eq:ABpsi}
    \Phi_{A,B}(\ketbra{\psi})[I] = \sum_{i\in I}\sum_{k\neq i} A_{ik}|\psi_k|^2\ketbra{i} + (B\odot \ketbra{\psi})[I].
\end{equation}
Note that $\operatorname{span}\{\ket{j}\}_{j\in I\setminus J}\subseteq \ker\sum_{i\in I}\sum_{k\neq i} A_{ik}|\psi_k|^2\ketbra{i}$ and equality holds whenever $\psi_i\neq 0$ for all $i\in [d]$. Clearly, since the $j^{th}$ column of $(B\odot \ketbra{\psi})[I]$ is nothing but $\overbar{\psi_j}\ket{B[I]_j} \odot \ket{\psi[I]}$, the column set $\{\ket{(B\odot \ketbra{\psi})[I]_j} \}_{j\in I\setminus J}\subseteq \C{r}$ is linearly dependent as well, so that the two matrices in the sum in Eq.~\eqref{eq:ABpsi} again have a non-trivial common kernel. Hence, $\Phi_{A,B}(\ketbra{\psi})[I]$ is singular. \\

\noindent\textbf{Case 2.} $\Phi=\Psi_{A,B}$ is a DUC channel with $A\succcurlyeq 0$, $B=B^\dagger$ and $A_{ij}A_{ji}\geq |B_{ij}|^2$. \\
In this case, since $B$ is not positive semi-definite, the previous approach would not work. Even so, let us still try to express $\Psi_{A,B}(\ketbra{e}{e})[I]$ as a sum of positive semi-definite matrices:
\begin{align}
    \Psi_{A,B}(\ketbra{e}{e})[I] &= \sum_{i\in I} \left( A_{ii} + \sum_{\substack{j\neq i \\ j\notin I}} A_{ij}\right) \ketbra{i} + \sum_{\substack{i<j \\ i,j\in I}}
    \big(A_{ij}\ketbra{i} + B_{ij}\ketbra{i}{j} + B_{ji}\ketbra{j}{i} + A_{ji}\ketbra{j}\big) \nonumber \\
&= \sum_{i\in I} \ketbra{\lambda_i} + \sum_{\substack{i<j \\ i,j\in I}} \ketbra{\lambda^1_{ij}} + \ketbra{\lambda^2_{ij}} \nonumber \\
&= \Lambda^\top \overbar{\Lambda},
\end{align}
where the vectors $\ket{\lambda_i}\in \operatorname{span}\{\ket{i}\}\subseteq \C{r}$ (for $i\in I$) and $\ket{\lambda^k_{ij}}\in \operatorname{span}\{\ket{i},\ket{j}\}\subseteq \C{r}$ (for $i<j$, $i,j\in I$, and $k\in\{1,2\}$) form the rank one decompositions of the corresponding positive semi-definite matrices from the previous sum, and $\Lambda\in \M{n\times r}$ stores these vectors in its rows. Notice that the conditions $A\succcurlyeq 0$, $B=B^\dagger$ and $A_{ij}A_{ji}\geq |B_{ij}|^2$ ensure that all the matrices in the two sums above are indeed positive semi-definite. Similar decomposition can be obtained for an arbitrary pure input state $\ketbra{\psi}\in\St{d}$:
\begin{alignat}{2}
    \Psi_{A,B}(\ketbra{\psi})[I] &= \sum_{i\in I} \left( A_{ii}|\psi_i|^2 + \sum_{\substack{j\neq i \\ j\notin I}} A_{ij}|\psi_j|^2\right) \ketbra{i} + \sum_{\substack{i<j \\ i,j\in I}} 
    \Big(&&\hspace{-2.6cm}A_{ij}|\psi_j|^2\ketbra{i}+B_{ij}\psi_j \overbar{\psi_i}\ketbra{i}{j} \nonumber \\[-0.7cm]
    & &&\hspace{-2cm}+B_{ji}\psi_i \overbar{\psi_j}\ketbra{j}{i} + A_{ji}|\psi_i|^2\ketbra{j}\Big ) \nonumber \\
&= \sum_{i\in I} \ketbra{\psi(\lambda)_i} + \sum_{\substack{i<j \\ i,j\in I }} \ketbra{\psi(\lambda)^1_{ij}} + \ketbra{\psi(\lambda)^2_{ij}} \nonumber \\
&= \Lambda_\psi^\top \overbar{\Lambda_\psi},
\end{alignat}
where the vectors $\ket{\psi(\lambda)_i}\in \operatorname{span}\{\ket{i}\}\subseteq \C{r}$ (for $i\in I$), $\ket{\psi(\lambda)^k_{ij}}\in \operatorname{span}\{\ket{i},\ket{j}\}\subseteq \C{r}$ (for $i<j$, $i,j\in I$, $k\in \{1,2\})$ and the matrix $\Lambda_\psi\in\M{n\times r}$ are defined as before. Notice that for all $i,j,k$:
\begin{equation}\label{eq:1}
    \ket{\psi(\lambda)^k_{ij}} = \ket{\lambda^k_{ij}} \odot (\psi_j\ket{i}+\psi_i \ket{j}),
\end{equation}
where $\odot$ denotes the entrywise vector product with respect to the standard basis in $\C{r}$. Now, if $n<r$, we have $\operatorname{rank}\Psi_{A,B}(\ketbra{e}{e})[I]=\operatorname{rank}\Lambda <r$ and $\operatorname{rank}\Psi_{A,B}(\ketbra{\psi})[I]=\operatorname{rank}\Lambda_\psi<r$, so that both $\Psi_{A,B}(\ketbra{e}{e})[I]$ and $\Psi_{A,B}(\ketbra{\psi})[I]$ are always singular. Otherwise, if $n\geq r$, we need to show that for $r\times r$ submatrices $\Lambda [J|r]:=\Lambda^{(r)}$ and $\Lambda_\psi[J|r]:=\Lambda^{(r)}_{\psi}$ constructed by arbitrarily choosing $r$ rows according to a random index set $J\subseteq [n]$ with $|J|=r$, 
\begin{equation}
    \det \Lambda^{(r)}=0 \implies \det \Lambda^{(r)}_{\psi} = 0.
\end{equation}
If $\det \Lambda^{(r)}=0$ because $\forall\pi\in S_r: \Lambda^{(r)}_\pi=0$, then since $\Lambda^{(r)}_{ij}=0\implies [\Lambda^{(r)}_{\psi}]_{ij}=0$ for all $i,j$, the desired claim is trivial to prove. Otherwise, there must exist at least two distinct permutations $\pi_1, \pi_2\in S_r$ such that $\Lambda^{(r)}_{\pi_1}$ and $\Lambda^{(r)}_{\pi_2}$ are non-zero. In this case, after suitable elementary row and column operations, $\Lambda^{(r)}$ and $\Lambda^{(r)}_{\psi}$ can be brought, respectively, into the block diagonal forms
\begin{equation}
    \left(
\begin{array}{ c c }
   X & 0  \\
   0 & Y
\end{array}
\right) \quad\text{and}\quad \left(
\begin{array}{ c c }
   X_{\psi} & 0  \\
   0 & Y_\psi
\end{array}
\right),
\end{equation}
where $X, X_\psi\in \M{k}$ are such that $\det X = X_\pi\neq 0$ and  $\det X_\psi= [X_\psi]_\pi$ for a unique $\pi\in S_k$ and $Y, Y_\psi = \widetilde{\psi}(Y)\in \M{r-k}$ are of the form described in Lemma~\ref{lemma:det}, see Eq.~\eqref{eq:1}. In the notation of Lemma~\ref{lemma:det}, the vector $\widetilde{\psi}\in \mathbb{C}^{r-k}$ which implements the transformation $Y\mapsto \widetilde{\psi}(Y)$ can be obtained by choosing some $r-k$ entries from $\psi\in \C{d}$. Its exact form is irrelevant for our purposes. When written in the above form, it should be evident that $\det \Lambda^{(r)}=0 \iff\det Y=0$, in which case Lemma~\ref{lemma:det} informs us that $\det Y_\psi = \det \widetilde{\psi}(Y)=0 \implies \det \Lambda^{(r)}_{\psi}=0$, and the proof is complete. \qed

\vspace{1cm}



\noindent\textbf{Acknowledgements.} The authors acknowledge support from the Cantab Capital Institute for the Mathematics of Information (CCIMI) and the Department of Applied Mathematics and Theoretical Physics (DAMTP), University of Cambridge.



\bigskip

\bibliography{references}
\bibliographystyle{alpha}

\vspace{0.5cm}
\hrule 
\vspace{0.5cm}

\end{document}